\documentclass[a4paper,11pt]{amsart} 
\usepackage [latin1]{inputenc}
\voffset - 2truecm
\hoffset - 1.4 truecm
\def\beq{\begin{equation}}
\def\eeq{\end{equation}}
\usepackage{amsmath,amsfonts,amssymb,amsthm}   
\usepackage{mathrsfs}
\usepackage{multicol}
\usepackage{fancybox} 
\usepackage{color}
\begin{document}
\def\Sym{{\rm Sym}}
\def\R{\mathbb R}
\def\N{\mathbb N}
\def\C{\mathbb C}
\def\Z{\mathbb Z}
\def\Q{\mathbb Q}
\def\bigat{\mbox {\Large (}}
\def\bigct{\mbox {\Large )}}
\def\Ker{{\cal N}}
\def\dom{{\rm dom}\;}
\def\supp{{\rm supp}\;}
\def\pf{\rm pf}
\def\arccot{{\rm arccot}}
\def\nequiv{\equiv \!\!\!\!\! / \,\, }
\def\tr{{\rm tr}\,}
\def\sgn{{\rm sgn}\,}
\def\vol{{\rm vol}}
\def\sym{{\rm sym}}
\def\skw{{\rm skw}}
\def\sp{{\rm sp}}
\def\dual{{\rm dual}}
\def\ctg{{\rm cot}}
\def\cot{{\rm cot}}
\def\sec{{\rm sec}}
\def\csc{{\rm csc}}
\def\diag{{\rm diag}}
\def\arcos{{\rm arccos}}
\def\bigaq{\mbox {\Large [}}
\def\bigcq{\mbox {\Large ]}}
\def\cA{{\mathcal A}}
\def\cB{{\mathcal B}}
\def\cC{{\mathcal C}}
\def\cD{{\mathcal D}}
\def\cE{{\mathcal E}}
\def\cF{{\mathcal F}}
\def\cG{{\mathcal G}}
\def\cH{{\mathcal H}}
\def\cX{{\mathcal X}}
\def\cI{{\cal I}}
\def\cL{{\mathcal L}}
\def\cN{{\cal N}}
\def\cO{{\cal O}}
\def\cR{{\mathcal R}}
\def\cP{{\cal P}}
\def\cS{{\cal S}}
\def\cT{{\cal T}}
\def\cU{{\cal U}}
\def\cV{{\cal V}}
\def\a{[\!\!\,[}
\def\c{]\!\!\,]}
\def\one{1\!\!\,{\rm l}}
\def\m{{\scriptscriptstyle -}}
\def\p{{\scriptscriptstyle +}}
\def\pp{{\scriptscriptstyle {++}}}
\def\sfot{{\scriptstyle{\frac 1 2}}}
\def\sfoth{{\scriptstyle{\frac 1 3}}}
\def\sfof{{\scriptstyle{\frac 1 4}}}
\def\sftth{{\scriptstyle{\frac 2 3}}}
\def\crux{\,{\times \!\!\!\!\! \times}\,}
\def\lbox{\mbox{\raisebox{-.25ex}{\Large$\Box$}}}
\def\boldeta{\mbox{\boldmath $\eta$}}
\def\bepsilon{\mbox{\boldmath $\epsilon$}}
\def\bgamma{\mbox{\boldmath $\gamma$}}
\def\bomega{\mbox{\boldmath $\omega$}}
\def\btheta{\mbox{\boldmath $\theta$}}
\def\bphi{\mbox{\boldmath $\phi$}}
\def\bpsi{\mbox{\boldmath $\psi$}}
\def\bxi{\mbox{\boldmath $\xi$}}
\def\bvarphi{\mbox{\boldmath $\varphi$}}
\def\bvarpi{\mbox{\boldmath $\varpi$}}
\def\bchi{\mbox{\boldmath $\chi$}}
\def\btau{\mbox{\boldmath $\tau$}}
\def\blambda{\mbox{\boldmath $\lambda$}}
\def\bkappa{\mbox{\boldmath $\kappa$}}
\def\bmu{\mbox{\boldmath $\mu$}}
\def\bnu{\mbox{\boldmath $\nu$}}
\def\bsigma{\mbox{\boldmath $\sigma$}}
\def\bvarepsilon{\mbox{\boldmath $\varepsilon$}}
\def\bPi{\mbox{\boldmath $\Pi$}}
\def\bPhi{\mbox{\boldmath $\Phi$}}
\def\bGamma{\mbox{\boldmath $\Gamma$}}
\def\bLambda{\mbox{\boldmath $\Lambda$}}
\def\bOmega{\mbox{\boldmath $\Omega$}}
\def\bXi{\mbox{\boldmath $\Xi$}}
\def\bUpsilon{\mbox{\boldmath $\Upsilon$}}
\def\bcA{\mbox{\boldmath ${\mathcal A}$}}
\def\bcB{\mbox{\boldmath ${\cal B}$}}
\def\bcS{\mbox{\boldmath ${\cal S}$}}
\def\bcU{\mbox{\boldmath ${\cal U}$}}
\def\bcE{\mbox{\boldmath ${\cal E}$}}
\def\bcE{\pmb{\mathcal E}}
\def\bcF{\pmb{\mathcal F}}
\def\bcG{\pmb{\mathcal G}}
\def\bcT{\pmb{\mathcal T}}
\def\bcB{\pmb{\mathcal B}}
\def\bcD{\pmb{\mathcal D}}
\def\bcJ{\pmb{\mathcal J}}
\def\bcY{\pmb{\mathcal Y}}
\def\bcQ{\pmb{\mathcal Q}}
\def\bcH{\pmb{\mathcal H}}
\def\bcK{\pmb{\mathcal K}}
\def\bcV{\pmb{\mathcal V}}
\def\cF{{\mathcal F}}
\def\cG{{\mathcal G}}
\def\cH{{\mathcal H}}
\def\cL{{\mathcal L}}
\def\baq{\mbox{\bf [}}
\def\bcq{\mbox{\bf ]}}
\def\bu{{\bf u}}
\def\bv{{\bf v}}
\def\bx{{\bf x}}
\def\bp{{\bf p}}
\def\bb{{\bf b}}
\def\bc{{\bf c}}
\def\bg{{\bf g}}
\def\bj{{\bf j}}
\def\bog{{\bf g}}
\def\bX{{\bf X}}
\def\bE{{\bf E}}
\def\bd{{\bf d}}
\def\be{{\bf e}}
\def\bi{{\bf i}}
\def\bl{{\bf l}}
\def\bQ{{\bf Q}}
\def\bY{{\bf Y}}
\def\bF{{\bf F}}
\def\bT{{\bf T}}
\def\bG{{\bf G}}
\def\bD{{\bf D}}
\def\bB{{\bf B}}
\def\bH{{\bf H}}
\def\bI{{\bf I}}
\def\bK{{\bf K}}
\def\bL{{\bf L}}
\def\bO{{\bf O}}
\def\bP{{\bf P}}
\def\bM{{\bf M}}
\def\bm{{\bf m}}
\def\bq{{\bf q}}
\def\bJ{{\bf J}}
\def\bC{{\bf C}}
\def\bR{{\bf R}}
\def\bS{{\bf S}}
\def\bW{{\bf W}}
\def\bw{{\bf w}}
\def\bZ{{\bf Z}}
\def\bz{{\bf z}}
\def\bn{{\bf n}}
\def\bN{{\bf N}}
\def\bs{{\bf s}}
\def\ba{{\bf a}}
\def\bt{{\bf t}}
\def\by{{\bf y}}
\def\br{{\bf r}}
\def\bh{{\bf h}}
\def\bk{{\bf k}}
\def\bo{{\bf o}}
\def\sq{{\scriptscriptstyle q}}
\def\sA{{\scriptscriptstyle A}}
\def\sB{{\scriptscriptstyle B}}
\def\sC{{\scriptscriptstyle C}}
\def\sD{{\scriptscriptstyle D}}
\def\sE{{\scriptscriptstyle E}}
\def\sF{{\scriptscriptstyle F}}
\def\sK{{\scriptscriptstyle K}}
\def\sL{{\scriptscriptstyle L}}
\def\sM{{\scriptscriptstyle M}}
\def\sN{{\scriptscriptstyle N}}
\def\sT{{\scriptscriptstyle T}}
\def\sG{{\scriptscriptstyle G}}
\def\sO{{\scriptscriptstyle O}}
\def\sP{{\scriptscriptstyle P}}
\def\sQ{{\scriptscriptstyle Q}}
\def\sR{{\scriptscriptstyle R}}
\def\sS{{\scriptscriptstyle S}}
\def\sbT{{\scriptscriptstyle \bT}}
\def\bsN{\mbox{\boldmath ${\mathsf N}$}}
\def\bsM{\mbox{\boldmath ${\mathsf M}$}}
\def\bsK{\mbox{\boldmath ${\mathsf K}$}}
\def\bsH{\mbox{\boldmath ${\mathsf H}$}}
\def\bsX{\mbox{\boldmath ${\mathsf X}$}}
\def\bsF{\mbox{\boldmath ${\mathsf F}$}}
\def\bsG{\mbox{\boldmath ${\mathsf G}$}}
\def\bsC{\mbox{\boldmath ${\mathsf C}$}}
\def\bsA{\mbox{\boldmath ${\mathsf A}$}}
\def\bsB{\mbox{\boldmath ${\mathsf B}$}}
\def\bsP{\mbox{\boldmath ${\mathsf P}$}}
\def\bsI{\mbox{\boldmath ${\mathsf I}$}}
\def\bsL{\mbox{\boldmath ${\mathsf L}$}}
\def\bsJ{\mbox{\boldmath ${\mathsf J}$}}
\def\bs0{\mbox{\boldmath ${\mathsf O}$}}
\def\sbcE{{\pmb{\scriptscriptstyle \mathcal E}}} 
\def\bU{{\bf U}}
\def\bV{{\bf V}}
\def\bof{{\bf f}}
\def\vp{\varphi}
\def\bSig{{\bf Sigma}}
\def\bPsi{{\bf Psi}\!\!\!\!\!@!@!{\bf Psi}}
\def\cvd{$\qquad\Box$}
\def\vcvd{\hfill $\sqcap \!\!\!\!\sqcup$}
\def\bone{{\bf 1}}
\def\bzero{{\bf 0}}
\def\curl{\nabla \crux}
\def\dive{\nabla \cdot}
\def\circo{\buildrel\circ\over}
\def\diamo{\buildrel\diamond\over}
\def\bA{{\bf A}}
\def\nonimplica{\Longrightarrow \!\!\!\!\!\!\!\! /\quad}
\def\para{{\scriptscriptstyle{\parallel}}}
\def\perpe{{\scriptscriptstyle{\perp}}}
\def\Res{{\rm Res}}
\def\sinc{{\rm sinc}}
\def\scrE{\mathscr{E}}
\def\scrT{\mathscr{T}}
\def\scrS{\mathscr{S}}
\def\rR{{\rm R}}
\def\nablaR{\nabla\!_{\sR}\,}
\def\mT{{\mathfrak{T}}}
\def\mE{{\mathfrak{E}}}
\def\mS{{\mathfrak{S}}}
\def\msfS{\mbox{${\mathsf S}$}}
\def\strianup{{\scriptscriptstyle{\bigtriangleup}}}
\def\striandown{{\scriptscriptstyle{\bigtriangledown}}}
\def\triangoloup{\buildrel\strianup\over}
\def\triangolodown{\buildrel\striandown\over}
\def\SSQ#1#2{ {\vbox{\hrule height#2pt\hbox{\vrule width#2pt height#1pt
    \kern#1pt \vrule width#2pt }\hrule height#2pt}}\smallskip\relax\rm}
\def\cvd{\SSQ{2.5}{.4}}
\def\quadro{\buildrel\cvd\over}
\def\quadrato{\buildrel\cvd\over}
\def\triangolidown{\buildrel{\striandown\striandown}\over}
\def\sBox{{\scriptscriptstyle{\Box}}}
\def\quadri{\buildrel{\cvd\,\cvd}\over}
\def\tcr{\textcolor{red}}
\newtheorem{remark}{Remark}[section]
\newtheorem{proposition}{Proposition}[section]
\newtheorem{theorem}[proposition]{Theorem}
\newtheorem{corollary}[proposition]{Corollary}
\newtheorem{lemma}[proposition]{Lemma}

\title[Nonlinear and nonlocal models of heat conduction]
{Nonlinear and nonlocal models of\\ heat conduction in  continuum thermodynamics}

\author[C. Giorgi]{Claudio Giorgi}%
\address{Universit\`{a} di Brescia , DICATAM
\newline \indent Via Valotti 9, 25133
Brescia,
Italia.}%
\email{claudio.giorgi@unibs.it}%


\author[F. Zullo]{Federico Zullo}%
\address{Universit\`{a} di Brescia , DICATAM
\newline \indent Via Valotti 9, 25133
Brescia,
Italia,
\newline \indent \& INFN, Milano Bicocca,
\newline \indent Piazza della Scienza 3, 20126, Milano, Italia}%
\email{federico.zullo@unibs.it}%

\subjclass [2000]{80}%

\begin{abstract}
 The aim of this paper is to develop a general constitutive scheme within continuum thermodynamics to describe the behavior of heat flow in deformable media. Starting from a classical thermodynamic approach, the rate-type constitutive equations are defined in the material (Lagrangian)  description where the standard time derivative satisfies the principle of objectivity. 
All constitutive functions are required to depend on a common set of independent variables and to be consistent with thermodynamics. The statement of the Second Law is formulated in a general nonlocal form, where the entropy production rate is prescribed by a non-negative constitutive function and the extra entropy flux obeys a no-flow boundary condition. 
The thermodynamic response is then developed based on Coleman-Noll procedure. In the local formulation, the free energy potential and the rate of entropy production function are assumed to depend on temperature, temperature gradient and heat-flux vector along with their time derivatives.
This approach results in rate-type constitutive equations for the heat-flux vector that are intrinsically consistent with the Second Law and easily amenable to analysis.
A huge class of linear and nonlinear models of the rate type are recovered
(e.g., Cattaneo-Maxwell's, Jeffreys-like, Green-Naghdi's, Quintanilla's and Burgers-like heat conductors). 
In the (weakly) nonlocal  formulation of the second law, both the entropy production rate and an entropy extra-flux vector  are assumed to depend on temperature, temperature gradient and heat-flux vector along with their spatial gradients and time derivatives. Within this (classical) thermodynamic framework the nonlocal Guyer-Krumhansl model and some nonlinear generalizations devised by Cimmelli and Sellitto are obtained.
 \end{abstract}
\thanks{The research leading to this work has been developed under the auspices of INDAM-GNFM. F.Z. acknowledges also the support of INFN, Gr. IV - Mathematical Methods in NonLinear Physic.}

\maketitle

\section{Introduction}
In the last years, numerous heat conduction models beyond Fourier have been developed to account for relaxational and nonlocal effects, fast phenomena or wave propagation, such as being typical for biological systems, nanomaterials or nanosystems. Non-Fourier models mainly differ for their various thermodynamic backgrounds (thermodynamics of irreversible processes, extended irreversible thermodynamics, etc.). A challenging question is their possible compatibility with the Second Law. In the book of Straughan \cite{Straughan} many of them are presented  and discussed in connection with wave propagation properties. A recent review by Kov\'acs \cite{Kovacs} discusses properties concerning their possible practical applications in light of experiments.
This article aims to discuss their deduction in the context of classical Continuum Thermodynamics and possible compatibility with the Second Law stated therein.

Fourier's law gives a macroscopic description of the microscopic phenomena associated with heat diffusion and is an excellent approximation at length scales much greater than the mean free path and at time scales much greater than the thermal relaxation time. Nevertheless, one of the predicted results of the Fourier law is that  temperature disturbances propagate at infinite speeds. This violates the law of special relativity, and also, since in metals the heat conduction is typically attributed to the migration of free electrons and in semiconductors it is explained by the migration of phonons, which are collective vibrations of the atoms in the crystal structure, it would be a reasonable requirement that the propagation speed proves to be finite. This view has spurred much academic interest in the last half century towards seeking a model compatible with a finite speed of propagation (see, e.g., \cite{Straughan,BL1995, Cattaneo48,GN2,GP,JP} and refs therein).  

The first non-classical heat conduction model capable of predicting the propagation of heat waves is the Maxwell-Cattaneo-Vernotte law (see \cite{Cattaneo48,Cattaneo58,Vernotte}).
The MCV theory is based on a
\emph{rate-type} constitutive equation for the heat flux that predicts heat-wave propagation.
Nevertheless, in the recent past some controversy raised about the non-objective character of
the constitutive equation which limits its application to rigid bodies at rest only.
Several efforts have been devoted to circumvent such a difficulty (cf.\ \cite{Ch, Morro_ORE, Morro_HC} and references therein).
Yet the MCV model suffer from some other drawbacks. In particular, the non-negativity of the absolute temperature at any time is not preserved (see, e.g., \cite{BZ2009,FGM_frac,Rukolaine}).

A new class of models for heat conduction in a rigid body  has been developed
in the nineties by Green and Naghdi \cite{GN}. In the framework of their general theory, the
propagation of thermal waves at finite speed is allowed. Unlike the MCV model, their constitutive equations are completely immune from criticism of  objectivity,
since do not contain the material derivative of any vector field.
Green and Naghdi proposed
three types of models, named \emph{type I, type II} and \emph{type
III}, respectively, the latter being the most general, which (formally) includes the others
as particular instances.
An empirical temperature scale is used, not necessarily the absolute one.

In the past decades, Green-Naghdi's theories have attracted growing interest. In particular, Green-Naghdi type III (GN III) theory  owes its success to the capability
of describing heat propagation by means of thermal waves in addition to diffusive propagation and for this reason it  has been applied in a number of disparate physical circumstances, where propagation of heat is coupled with elastic deformations of solids, flow of viscoelastic fluids, etc. (see \cite{Straughan,GN93,QS} and references therein).
On the other hand, some criticisms have been raised about the GN III heat conduction theory \cite{Bargmann,BFPG,BFPG0}. Mainly, it has not been demonstrated that the internal dissipation of entropy (which is assumed to be the subject of a constitutive prescription) is non-negative, as required by the Second Law. Moreover, it has been observed that the GN III model fits well to propagation processes of finite duration (transient regime), but leads to unrealistic effects, at least when either asymptotic or stationary phenomena are involved \cite{GGP}.

Recently, Quintanilla \cite{Quintanilla} proposed a new thermoelastic heat conduction model which emerged from the development of the GN III by adding a relaxation factor. This theory gives rise to a third-order differential equation for the temperature that looks like the linearized Moore-Gibson-Thompson (MGT) equation in high intensity ultrasound \cite{Thompson}. A similar equation had already been obtained by Joseph and Preziosi starting from the linearization of heat conduction with memory according to the Gurtin-Pipkin model (see \cite[eqn.(5.7)]{JP}). 
Despite the advantages and wide diffusion of this new theory \cite{Quintanilla1,Abou,BFQ,SM1,SM}, a convincing proof of its thermodynamic consistency is currently lacking.

Guyer and Krumhansl \cite{GK} studied the heat wave propagation in dielectric crystals at low temperature. They observed that in the regime of low temperature the heat flux $\bq$ is proportional to the momentum flux of the phonon gas. On the basis of kinetic theory they found a macroscopic equation governing its evolution. 
When the relaxation time is negligible,  the Guyer-Krumhansl (GK) equation  reduces to a nonlocal perturbation of the Fourier law  (see \cite{GK2}). 
The possibility of utilizing a continuum approach to derive the Guyer-Krumhansl equation results in a model with a wide range of validity since it allows to fit the parameters to the observed phenomenon.
Considering an undeformable medium, the GK equation has been derived within the framework of classical irreversible thermodynamics with internal variables \cite{CR,CR1},  GENERIC \cite{SPKFVG} and extended irreversible thermodynamics \cite{JCVL}. 
A simple nonlinear extension of the GK model was proposed and scrutinized in \cite{CSJ_2009,CSJ_2010,CSJ_2010bis}. Such a model illustrates relaxational and nonlocal effects of the heat flow in nanosystems. We are not aware whether the GK model and its  consistency with the Second Law have been obtained in the context of classical continuum thermodynamics. Moreover, it is not entirely clear which boundary conditions should be assigned for the heat flux.

\subsection{Aims and plan of the paper}
In this work we propose a new approach to heat conduction theories that is inherently thermodynamic, as it originates directly from the Clausius-Duhem inequality. 
The {\it specific production of entropy} $\sigma$ enters as a non-negative constitutive function, so that the Second Law is automatically satisfied.
The idea that  $\sigma$ be given by a constitutive equation traces back to Green and Naghdi \cite{GN}. 
They assumed from the outset that all constitutive functions depend on the common set of (physical) variables. 
Conceptually our approach extends this idea. 

The set of independent variables includes only the macroscopically observable fields and their temporal and spatial derivatives, without making any recourse to internal variables or ambiguous state variables (such as thermal displacement).  In particular, no constitutive prescription on the energy influx $\bq$ is made, rather it is treated as an independent variable. The Coleman-Noll procedure \cite{CN63} is applied to derive thermodynamic restrictions on  the Helmholtz free energy $\psi$. However, unlike the Green-Naghdi theories, here $\sigma$ is independent of the constitutive prescription of  the free energy $\psi$ and, what is more, it is proved to be non-negative along whatever process.  
 
 This strategy allows us to increase the generality of the models considered here, for instance new classes of rate-type constitutive equations for the heat flux.
Since time derivatives are involved, we adopt a material description in order to avoid the problem of their objectivity. Therefore, this approach is compatible with both rigid and deformable bodies.
We mention that a similar approach was developed in \cite{MCS2,CMT} for deformable ferroelectrics, in \cite{MorroGiorgi_JELAS} for elastic-plastic materials and in \cite{MorroGiorgi_MDPI, GM_Materials} for viscoelastic and viscoplastic materials. 

We consider both local and (weakly) nonlocal models of heat conduction. 
Constitutive local theories  have to satisfy the classical local formulation of the Second Law where
the specific entropy production $\sigma$ is  prescribed as a constitutive function in terms of the independent variables and is required to be non-negative in all processes. 
Depending on the choice of the set of variables (especially the order of the time derivatives of $\bq$ involved) we consider different classes of rate-type models (basic rather than first, second or higher  order). Among others, we obtain GNIII and Quintanilla models in their differential form, and their thermodynamic consistency is shown by exhibiting a (non-unique) explicit expressions of $\psi$ and $\sigma$.  Then, a more general heat flow model inspired by the Burgers fluid is proposed. Its thermodynamic consistency, in terms of inequalities involving the material parameters of the model, is proved by assuming a quadratic form of the free energy.

As for non-local theories, we refer to Green and Laws \cite{GrL} who first investigated the restrictions imposed on constitutive equations by a global entropy inequality. According to \cite{FGM_PhysicaD} we reformulated the Green-Laws pioneering idea by arriving at an entropy equation in which the entropy production density $\rho\sigma$ is the sum of two terms; an internal entropy production, $\rho\zeta$, which is required to be non-negative along whatever process, and the divergence of a vector field $\bk$, called {\it extra entropy flux}, whose flow across the boundary of the body is zero. Both  $\zeta$ and $\bk$ are the object of a constitutive prescriptions in terms of the same independent variables as  $\psi$. Among others, linear and nonlinear Guyer-Krumhansl-like heat conduction models are derived using this scheme and their thermodynamic consistency is shown by exhibiting an explicit expressions of $\psi$, $\zeta$ and $\bk$. Restrictions imposed on $\bq$ and its gradient by the no-flow boundary condition for $\kappa$ are discussed. In our opinion these conditions provide an important suggestion for the choice of the most appropriate boundary conditions in applications.
Finally, a connection with the evolution equation of the heat flux in the framework of extended irreversible
thermodynamics (see \cite{JCVL}) is established.

\section{Balance laws and the thermodynamic principles}

We consider a body occupying the time dependent region $\Omega_t \subset \cE^3$.
The motion is described by function $\bchi(\bX, t)$ providing the position vector $\bx \in \Omega_t$ in terms of the position vector $\bX$, in a reference configuration $\cR$, and the time $t$, so that 
\[\bx=\bchi(\bX, t),\qquad  \Omega_t = \bchi(\cR, t).\]
The deformation is described by means of the deformation gradient 
\[\bF (\bX, t)=\nablaR  \bchi(\bX, t),\qquad (\hbox{in suffix notation }F_{iK} = \partial_{X_K}\chi_i),\]
satisfying the constraint $J:=\det \bF>0$.
Here, $\nablaR:=\partial_\bX $ denotes the gradient in the reference configuration $\cR$, whereas the symbol $\nabla:=\partial_\bx$ denotes the gradient in the current configuration $\Omega$. For any regular vector field $\bw(\bx,t)$, they are related  as follows
 \[\nablaR\hat\bw=\nabla\bw\bF, \qquad \nabla\bw=\nablaR\hat\bw\bF^{-1},
\]
  where $\hat\bw=\bw\big(\bchi(\bX,t),t\big)$. In addition, using the Nanson's formula, we have
 \beq \nablaR\cdot\hat\bw=J\nabla\cdot\left(J^{-1} \bF\bw\right), \qquad \nabla\cdot\bw=J^{-1}\nablaR\cdot(J \bF^{-1}\hat\bw). \label{nabla_rel}\eeq

  Hereafter, a superposed dot denotes the standard derivative with respect to time. In particular, 
 \[\dot f(\bX,t):=\partial_t f(\bX,t),\qquad \dot g(\bx,t):=\frac{\,d}{dt}g(\bchi(\bX,t),t)=\partial_t g(\bx,t)+ (\bv\cdot\nabla)g(\bx,t),\]
 where $f$ and $g$ are scalar-, vector- or tensor-valued differentiable functions of the reference and current position, respectively. The velocity field $\bv(\bx,t)$ is such that $\bv(\bchi(\bX,t),t)=\partial_t\bchi(\bX,t)$ and the velocity gradient $\bL:=\nabla\bv$ is related to $\dot\bF$ as follows
  \beq \bL=\dot\bF\bF^{-1}.
 \label{vel_grad} 
 \eeq

Let $\varepsilon$ be the internal energy density (per unit mass), $\bT$  the Cauchy stress, $\bq$ the heat flux vector, $\rho$ the mass density, $r$ the (external) heat supply and $\bb$ the mechanical body force per unit mass. 
The local form of the linear momentum and internal energy balance equations can be written as
\beq
\rho \dot{\bv} = \nabla \cdot \bT + \rho \bb ,
\label{eq:1motion}\eeq
\beq 
\rho \dot{\varepsilon} =\bT \cdot \bL - \nabla \cdot \bq + \rho r.\label{eq:en} \eeq
The total energy balance, also named  First Law of continuum thermodynamics, follows from taking the dot product of  \eqref{eq:1motion} by $\bv$ and then adding the result to \eqref{eq:en},
\[\rho\big[\tfrac12|\bv|^2+{\varepsilon}\big]^\cdot = - \nabla \cdot \big[\bq -\bT\bv\big] + \rho[\bb\cdot\bv+ r].\]

\medskip\noindent
{\bf Balance of entropy.}
{\it Let  ${\mathcal P}_t$ be any convecting subregion of $\Omega_t$.
 All processes which are compatible with the balance equations \eqref{eq:1motion}-\eqref{eq:en} must satisfy the following integral equation, }
\[
	\frac{d}{dt} \int_{{\mathcal P}_t} \rho\, \eta \,dv
	 =  -\int_{\partial{\mathcal P}_t} {\bh}\cdot \bn \,da
	 +  \int_{{\mathcal P}_t} \rho s\, dv,
\]
{\it where $\eta$ is the entropy density and $s$ the entropy supply (per unit mass), ${\bh}$ the entropy-flux vector, $\bn$ the outward normal to the boundary.}

\smallskip
This general statement has been adopted by M\"uller \cite{Muller85} by allowing the entropy flux $\bh$ to be an unknown function.
The classical form of the Second Law, usually named after Clausius-Duhem, is obtained by letting
$$
{\bh}= \frac{\bq}{\theta},\qquad s=\frac{r}{\theta}+\sigma,
$$
where $\theta$ denotes the (positive) absolute temperature and the quantity $\sigma$ is referred to as {\it specific entropy production} \cite{CN63,Muller71}.
Exploiting the arbitrariness of the convecting domain ${\mathcal P}_t$, we obtain the local form of the entropy balance,
\begin{equation}
\label{loc_eta_ineq}
\rho \dot{\eta}  + \nabla\cdot\left( \frac{\bq}{\theta}  \right)- \frac{\rho r}{\theta}=\rho\sigma.
\end{equation}

Within continuum thermodynamics two versions of the Second Law  can be established depending on the assumptions about $\sigma$.

\medskip\noindent
{\bf Second Law of Thermodynamics (local form).} {\it Along all compatible processes the specific entropy production $\sigma$  is nonnegative,}
\beq
\sigma(\bx,t)\ge0.
\label{strong_ineq}\eeq
Owing to \eqref{loc_eta_ineq} and \eqref{strong_ineq}, this statement leads to  the classical form of the Clausius-Duhem inequality 
\[\rho \dot{\eta}  + \nabla\cdot\left( \frac{\bq}{\theta}  \right)- \frac{\rho r}{\theta}\ge0.\]

A more general formulation was introduced and discussed in  \cite{GrL}; it can be summarized as follows.

\medskip\noindent
{\bf Second Law of Thermodynamics (nonlocal weak form).} {\it The specific entropy production  $\sigma$ is 
 assumed to satisfy the integral condition}
\begin{equation}
\label{II_law} 
 \int_{\Omega_t} \rho\, \sigma \,dv\ge0.
\end{equation}
This statement  is weaker than \eqref{strong_ineq} because the inequality is assumed to be valid only over the entire domain $\Omega_t$  and not point-wise. 

Now, let  $\zeta$ be a nonnegative unknown field and $\bk $ be an unknown regular field, usually called {\it entropy extra flux}, such that
\[
\int_{\Omega_t}\nabla \cdot \bk \,dv =  \int_{\partial \Omega_t} \bk \cdot \bn\, da = 0,
\]
which means that the net extra flux across the whole body boundary is zero. This condition is also relevant from a physical point of view since the flux of $\bk$ at the boundary is not an observable quantity.
Then, condition \eqref{II_law} can be satisfied by letting
\begin{equation}
\rho \sigma=\rho\zeta-\nabla \cdot \bk,\qquad  \int_{\Omega_t} \rho\, \zeta(\bx,t) \,dv\ge0.
\label{zeta_k}\end{equation}
We stress that equation \eqref{zeta_k}$_1$  could be  interpreted as a microscopic local balance at any point $\bx\in\Omega_t$, where the entropy production density $\sigma$  is due to a {\it local entropy supply} $\zeta$ and an internal entropy flux $\bk$ that is exchanged on a microscopic level with the surrounding points. The idea to introduce the extra entropy flux traces back to and is discussed in \cite{FGM_PhysicaD}.

Hereafter, the validity  of the Second Law is ensured by the following stronger formulation.

\smallskip\noindent
{\bf Second Law of Thermodynamics (nonlocal strong form).} {\it During all compatible processes the entropy extra flow at any point of the body boundary is zero and the local entropy supply at any point of the body interior is nonnegative, namely}
\begin{equation}
\label{strong_II_law} 
\bk\cdot\bn\big\vert_{\partial\Omega_t}=0,\quad \zeta\ge0 \quad \hbox{in }\Omega_t.
\end{equation}

\smallskip

Consequently, after replacing \eqref{zeta_k}$_1$ into \eqref{loc_eta_ineq}  it follows
\begin{equation}
\label{eta_ineq_extra}
\rho \dot{\eta}  + \nabla\cdot\left( \frac{\bq}{\theta} +\bk \right) - \frac{\rho r}{\theta}=\rho\zeta\ge0,
\end{equation}
that resembles the classical entropy inequality except that the  flux vector is redefined by adding the extra contribution $\bk$. This suggests that $\bk$ is a characteristic nonlocal contribute and it must disappear  when only local processes are involved.

  Henceforth, we assume that $\bk$ and $\zeta$ are given as constitutive functions of the assigned set of variables, as well as the internal energy and the entropy. The idea that the entropy production $\sigma$ be given by a constitutive equation traces back to Green and Naghdi 
\cite{GN}, however they do not require that $\sigma\ge0$ along whatever process. This idea was extended in \cite{CMT,MorroGiorgi_JELAS,MorroGiorgi_MDPI,GM_Materials} assuming from the outset that all constitutive functions depend on the common set of physical variables and imposing $\sigma\ge0$. Conceptually our contribution follows the same scheme (see also \cite{GM_book}).

Upon substitution of $\nabla\cdot \bq - \rho r$ from the energy equation (\ref{eq:en}) into \eqref{eta_ineq_extra} and multiplication by $\theta$ we obtain the basic thermodynamic relation
\beq - \rho(\dot{\psi} + \eta \dot{\theta})  + \bT \cdot \bL 
- \frac 1 \theta \bq \cdot \nabla\theta + \theta \nabla\cdot \bk =\rho\theta\zeta,\label{eq:CD0} \eeq
where $\psi= \varepsilon - \theta \eta$ denotes the  {\it Helmholtz free energy density}. 
Within the isothermal setting, the term $\rho\theta\zeta$ is usually referred to as {\it rate of (mechanical) dissipation} \cite{RS1999}. 
Due to \eqref{strong_II_law}, equation \eqref{eq:CD0} becomes an inequality that must be satisfied along whatever process.

Finally, multiplying \eqref{eq:CD0} by $J$ and using the identities \eqref{nabla_rel}, \eqref{vel_grad} and $\rho_\sR=J\rho$,
 we obtain the the basic thermodynamic inequality in the material description
 \beq\begin{split}
-{\rho_\sR}\big(\dot{\psi} + \eta \dot{\theta}
\big)
+  \bT_{\sR\sR} \cdot \dot{\bE} 
- \frac 1 {\theta} \bq_\sR \cdot \nablaR \theta
+ \theta\nablaR \cdot \bk_\sR =\rho_{\sR}\theta\zeta\ge0,
\end{split}\label{basic_ent_ineq}\eeq
where $\nablaR:=\bF^T\nabla$ and
\[ \bT_{\sR\sR} := J \bF^{-1} \bT \bF^{-T}, \qquad\bq_\sR := J \bF^{-1} \bq,\qquad\bk_\sR := J \bF^{-1} \bk, \quad
\bk_\sR\cdot\bn_\sR\big\vert_{\partial\cR}=0.
\]


\section{Local theories of heat conduction}

In the first part of this paper we restrict our attention to {\it local} theories, that is we assume $\bk=0$ and then $\sigma=\zeta\ge0$. 
Consequently \eqref{basic_ent_ineq} reduces to
\beq-{\rho_\sR}(\dot{\psi} + \eta \dot{\theta}) 
+ \bT_{\sR\sR} \cdot \dot{\bE} 
- \frac 1 {\theta}\, \bq_\sR \cdot \nablaR \theta=\rho_\sR\theta\sigma  \ge 0.\label{eq:CD} \eeq
In recent papers  \cite{MorroGiorgi_JELAS,MorroGiorgi_MDPI,GM_Materials} several local models of viscoelastic, viscoplastic and elastic-plastic materials (solids and fluids) are developed on the basis of \eqref{eq:CD} by specifying three elements;
\begin{itemize}
\item[-] the set $\Xi_\sR$ of admissible variables,
\item[-] the free energy density function $\psi=\psi(\Xi_\sR)$,
\item[-] the entropy production function $\sigma=\sigma(\Xi_\sR)$.
\end{itemize}
Those papers mainly scrutinized the mechanical characteristics of materials deriving quite general stress-strain constitutive equations from the Clausius-Duhem inequality.

Likewise, thermodynamically consistent constitutive equations of the rate type are devised here for heat-conduction in solids. The key idea is to exploit  the formal similarity of the scalar products
$$\bT_{\sR\sR} \cdot \dot{\bE},\qquad -\bq_\sR \cdot\nablaR(\ln \theta)$$
that represent the mechanical and thermal powers of internal forces, respectively. In \cite{MorroGiorgi_MDPI, GM_Materials} the exploitation of the thermodynamic inequality involving the mechanical power leads to some well-known viscoelastic models. By mimicking such a procedure, the exchange of $ \bT_{\sR\sR}$ with $\bq_\sR$, as well as that of $\dot\bE$ with $-\nablaR(\ln \theta)$, give rise to as many models of heat conduction.  

\subsection{Basic local models}

The simplest local models can be obtained by choosing a small set of independent variables, namely the temperature field and its gradient along with a strain measure and its rate. Since invariance requirements demand that the dependence on time and space derivatives occurs in an objective way, we let 
\[ \Xi_\sR := (\theta, \bE, \nablaR \theta,\dot \bE)\]
be the basic set of  invariant variables. Accordingly, Euclidean invariance of $\psi, \eta$ and $\sigma$ implies that their dependence  be a function of  $\Xi_\sR$.
Here and below, we assume $\eta$ is continuous while $\psi$ is continuously differentiable.

Upon evaluation of $\dot{\psi}$ and substitution in (\ref{eq:CD}) we obtain
\[  \rho_{\sR}(\partial_\theta \psi + \eta) \dot{\theta}
+ (\rho_{\sR} \partial_\bE \psi - \bT_{\sR\sR})\cdot \dot{\bE}  + \rho_{\sR} \partial_{\nablaR \theta} \psi \cdot \nablaR \dot{\theta}+ \rho_{\sR} \partial_{\dot\bE} \psi \cdot \ddot{\bE} 
 +\frac 1 \theta \bq_{\sR} \cdot \nablaR \theta=-\rho_\sR\theta\sigma.\] 
Due to the constitutive assumption on $\psi, \eta$ and $\sigma$, this expression depends linearly on $\dot{\theta}$, $\ddot{\bE} $ and $\nablaR \dot{\theta}$. Hence, their arbitrariness implies that
\[ \psi = \psi(\theta, \bE), \qquad \eta = - \partial_\theta \psi,\qquad   (\rho_{\sR} \partial_\bE \psi - \bT_{\sR\sR})\cdot \dot{\bE}  +  \frac1\theta \bq_{\sR} \cdot \nablaR \theta =-\rho_\sR\theta\sigma\le 0\]

If $\sigma$ is independent of $\dot\bE$ then the linearity and arbitrariness of $\dot{\bE} $ give
$$ \bT_{\sR\sR}=\rho_{\sR} \partial_\bE \psi, \qquad     \bq_{\sR} \cdot \nablaR \theta =-\rho_\sR\theta^2\sigma\le 0.$$ 
As is well known, this occurs in heat conducting hyperelastic materials. 

Otherwise, we split the last inequality as follows\footnote{Assuming that $\bT_{\sR\sR}$ is independent of $\nablaR\theta$ and $ \bq_\sR$ is independent of $\dot{\bE}$ , as usually happens, then $\sigma^{\sE\sT}=\sigma\vert_{\nablaR\theta=\bzero}$ and $\sigma^q=\sigma\vert_{\dot\bE=\bzero}$.} \cite{GM_Materials}
 \beq (\bT_{\sR\sR} - \rho_\sR \partial_\bE \psi)\cdot \dot{\bE}= \rho_{\sR}\theta \sigma^{\sE\sT}\ge0, \quad
\bq_\sR \cdot \nablaR \theta = -\rho_{\sR} \theta^2 \sigma^q\le0. \label{eq:TqR} \eeq
For example, it is easy to check that the former inequality  \eqref{eq:TqR}$_1$ is satisfied  by the Kelvin-Voigt viscoelastic model, 
$$\bT_{\sR\sR} = \bsC\bE + \bsK \dot{\bE},$$
and thermodynamic consistency is ensured by letting
\beq \rho_\sR\psi=\frac12\bE\cdot \bsC\bE,\qquad 
\rho_\sR\theta\sigma^{\sE\sT} =\dot{\bE}\cdot\bsK \dot{\bE},\label{eq:sigma_F0}\eeq
where $\bsK$ denotes a positive semidefinite fourth-order tensor.
On the other hand, the latter inequality  \eqref{eq:TqR}$_2$
 is satisfied by assuming Fourier's law for the heat flux vector, 
 $$\bq_\sR = - \bkappa \nablaR \theta.$$ In this case $\psi$ is not involved and thermodynamic consistency requires
\beq \rho_{\sR}\sigma^q = \nablaR(\ln \theta) \cdot \bkappa \nablaR(\ln \theta),\label{eq:sigma_F}\eeq
where $\bkappa$ denotes a positive semidefinite second-order tensor. 

A comparison of the expressions of $\sigma^{\sE\sT}$ and $\sigma^q$ as given by \eqref{eq:sigma_F0} and \eqref{eq:sigma_F} respectively, provides immediate evidence of our key idea.

In order to derive more general constitutive equations for $\bT_{\sR\sR} $ and $\bq_\sR$ we take advantage of the following property (see, {\it e.g.},  \cite[Proposition 1]{GM_Materials}). 

\medskip\noindent{\bf Representation Lemma.}
\it Given a vector (or a tensor)  $\bA$, let $\bN = \bA/|\bA|$. If $\bZ$ is a vector (or a tensor) such that only its projection on $\bN$ is known, namely
$$\bZ \cdot \bN=g,$$  
then for any vector (or tensor) $\bG$ we can write
\beq \bZ = g\bN + (\bone - \bN\otimes \bN)\bG 
\label{eq:BNG} \eeq
where  $\bone$ is the second-order (or fourth-order) identity tensor.

\rm\smallskip
Let $\bsM$ be a symmetric non-singular fourth-order tensor. After applying the Representation Lemma to \eqref{eq:TqR}$_1$ with
 $$\bN = \bsM \dot{\bE}/|\bsM \dot{\bE}|,\qquad
\bZ = \bsM^{-1}[\bT_{\sR\sR} - \rho_\sR \partial_\bE \psi],\qquad \bG =\bzero,$$
  we obtain the nonlinear constitutive equation
\[ \bT_{\sR\sR} = \rho_\sR \partial_\bE \psi +  \rho_\sR\theta \sigma^{\sE\sT} \frac{\bsM^2 \dot{\bE}}{|\bsM \dot{\bE}|^2}
 .\]
In particular, the linear Kelvin-Voigt relation follows from \eqref{eq:sigma_F0} with  the  special choice $\bsK=\bsM^2$.

Let $\bA$ be a any symmetric, non-singular, second-order tensor, possibly parametrized by the temperature  $\theta$. Applying (\ref{eq:BNG}) with $\bN = \bA \nablaR \theta/|\bA \nablaR \theta|$, $\bZ = \bA^{-1}\bq_\sR$, and $\bG = \bzero$,  from  \eqref{eq:TqR}$_2$ it follows that 
\[\bq_\sR = -  \frac{\rho_\sR\theta^2 \sigma^q}{|\bA \nablaR \theta|^2}\bA^2 \nablaR \theta. \]
In particular, assuming $\bkappa = \bA^2$ and $\sigma^q$ as in \eqref{eq:sigma_F}, we recover Fourier's law.  Note however that in this case $\bkappa$ turns out to be strictly positive definite because $\bA$ is non-singular.

\section{Rate-type local  models}

To describe rate-type models we expand the basic set of Euclidean invariant variables by adding some quantities that are usually not assumed to be independent, but considered as constitutive functions; specifically, $\bT_{\sR\sR}$ and $\bq_\sR$.  Hence we let
\[ \Xi_\sR := (\theta, \bE, \bT_{\sR\sR}, \bq_\sR, \nablaR \theta, \dot{\bE})\]
be the set of independent variables and assume that $\psi$, $\eta$, $\sigma$ are scalar-valued functions of $\Xi_\sR$.
 In view of the introduction of rate-type models,  we look for a scheme where $\dot \bq_\sR$ and $\nablaR \theta$, as well as $ \dot\bT_{\sR\sR}$ and $\dot{\bE}$, are regarded as mutually dependent variables\footnote{A similar approach occurs in anholonomic system described by a set of independent parameters subject to differential constraints that make their rates mutually dependent.}.

Upon evaluation of $\dot{\psi}$ and substitution in (\ref{eq:CD}) we obtain
\[ \begin{split} \rho_{\sR}(\partial_\theta \psi + \eta) \dot{\theta}
+ (\rho_{\sR} \partial_\bE \psi - \bT_{\sR\sR})\cdot \dot{\bE} + \rho_{\sR} \partial_{\bT_{\sR\sR}} \psi \cdot \dot{\bT}_{\sR\sR} + \rho_{\sR} \partial_{ \bq_\sR}\psi \cdot { \dot\bq_\sR}& \\
+ \rho_{\sR} \partial_{\nablaR \theta} \psi \cdot \nablaR \dot{\theta}  
+ \rho_{\sR} \partial_{\dot{\bE}}\psi \cdot \ddot{\bE}
 +\frac 1 \theta \bq_{\sR} \cdot \nablaR \theta=&-\rho_\sR\theta\sigma.\end{split} \] 
The linearity and arbitrariness of $\dot{\theta}$, $\ddot{\bE}$, 
$\nablaR \dot{\theta}$, imply that 
\beq \psi = \psi(\theta, \bE, \bT_{\sR\sR},\bq_\sR), \qquad  \eta=-\partial_\theta \psi. \label{eq:psi_first_ord}\eeq
and  the thermodynamic inequality reduces to 
\beq  (\rho_{\sR} \partial_\bE \psi - \bT_{\sR\sR})\cdot \dot{\bE} + \rho_{\sR} \partial_{\bT_{\sR\sR}} \psi \cdot \dot{\bT}_{\sR\sR} + \rho_{\sR}\partial_{ \bq_\sR}\psi \cdot { \dot\bq_\sR}
+ \frac {\bq_{\sR}} \theta \cdot \nablaR \theta=-\rho_\sR\theta\sigma \le 0. \label{eq:ine0}\eeq 
We can say that the arguments of $\psi$ in  \eqref{eq:psi_first_ord} provide the set of {\it state variables},
$$\Sigma_\sR=(\theta, \bE, \bT_{\sR\sR},\bq_\sR).$$

Now, let $\dot{\bE}$ and $\dot{\bT}_{\sR\sR}$ be independent of $\dot{\bq}_\sR, \nablaR \theta$. Since $\psi$ is also independent of $\dot{\bq}_\sR, \nablaR \theta$,  then letting $\dot\bq_\sR=\nablaR \theta=\bzero$ we can  write \eqref{eq:ine0} in the form
\beq  (\rho_{\sR} \partial_\bE \psi - \bT_{\sR\sR})\cdot \dot{\bE} + \rho_{\sR} \partial_{\bT_{\sR\sR}} \psi \cdot \dot{\bT}_{\sR\sR}=-\rho_\sR\theta\sigma^{\sE\sT} \le 0, 
 \label{eq:ineETq}\eeq
where $\sigma^{\sE\sT}$ is the entropy production density $\sigma$ when $\dot\bq_\sR=\nablaR\theta = \bzero$. 
Likewise,  assuming that $\dot{\bq}_\sR, \bq_\sR$, and $\nablaR \theta$ are independent of $\dot{\bE}, \dot{\bT}_{\sR\sR}$, as usually happens, then 
\beq 
\rho_{\sR}\partial_{ \bq_\sR}\psi \cdot { \dot\bq_\sR}
+ \frac {\bq_{\sR}} \theta \cdot \nablaR \theta=-\rho_\sR\theta\sigma^{q}\le 0,
\label{eq:ineETqq}\eeq 
where $\sigma^{q}$ is the entropy production density  when $\dot{\bT}_{\sR\sR}=\dot{\bE} = \bzero$. 
Apparently, $\sigma = \sigma^{\sE\sT} + \sigma^{q}$.  
The entropy productions $\sigma^{\sE\sT}$ and $\sigma^{q}$, as well as $\sigma$, are nonnegative constitutive functions of $\Xi_\sR$ to be determined according to the constitutive model.

 By exploiting inequality (\ref{eq:ineETq}) and disregarding heat conduction, memory properties of viscoelasticity, elastoplasticity and viscoplasticity were modeled with suitable nonlinear rate-type stress-strain relations  \cite{MorroGiorgi_JELAS,MorroGiorgi_MDPI}. 
However, since the aim of this paper is to establish nonlinear rate-type models of heat conduction, we limit our attention to (\ref{eq:ineETqq}) and disregard  (\ref{eq:ineETq}).
In detail, the dependence of constitutive functions  on ${\bE}, {\bT}_{\sR\sR}$ is neglected so that \eqref{eq:psi_first_ord} reduces to $\psi=\psi(\theta,\bq_\sR)$.

Although inequality  (\ref{eq:ineETqq}) is common to many approaches where both $\bq_\sR$ and $\nablaR \theta$ are independent variables  \cite{MCM, MorroJE} our scheme has the advantage that the material time derivative is objective \cite{MorroGiorgi} and makes consistency with thermodynamics much easier than it happens when heat conduction involves histories \cite{Meccanica}, summed histories \cite{GP} or internal variables \cite{Van2015,Maugin}.

Since (\ref{eq:ineETqq}) holds for any pair of functions $\psi$ and $\sigma^q\ge0$
of the variables $\Xi_{\sR}$, below we present some examples by choosing special expressions for these functions. Many of them match with well-known  thermal conduction models.

\subsection{Green-Naghdi type II  heat conductors}
A (linear) type II heat conductor according to Green and Naghdi \cite{GN} is characterized by
\beq
\bq_\sR=-\bxi\nablaR \alpha,
\label{GNII_original}\eeq
where $\alpha$ is an internal variable, named thermal displacement, that represents, by definition, a time primitive of the temperature, namely $\dot\alpha=\theta$. Since this relation enters the heat equation for the temperature $\theta$ only after a differentiation with respect to time, namely 
\[\dot\bq_\sR=-\bxi\nablaR \theta,\]
it is reasonable to consider this GN II constitutive equation  instead of \eqref{GNII_original}. Such a rate-type model can be obtained from inequality \eqref{eq:ineETqq} by appealing to the Representation Lemma. 

Assuming that $\sigma^{q}=0$ eqn.\,\eqref{eq:ineETqq} becomes
\beq  \rho_{\sR}\partial_{ \bq_\sR}\psi \cdot { \dot\bq_\sR} =- \frac{1}{\theta} \bq_{\sR}\cdot \nablaR \theta.
\label{Nondissipative_cond} \eeq
This condition characterizes heat conductors without entropy intrinsic production, that means without any dissipation.
An explicit relation for $\dot{\bq}_\sR$ then follows by  applying \eqref{eq:BNG} with 
\[\bN=\partial_{\bq_\sR} \psi/|\partial_{\bq_\sR} \psi|, \qquad\bZ={{\dot\bq}_\sR}, \qquad\bG=\bJ_\sR\nablaR \theta,\]
 where $\bJ_\sR = \bJ_\sR(\theta,\bq_\sR)$  is an arbitrary second-order tensor-valued function. Using \eqref{Nondissipative_cond} and assuming $ \partial_{ \bq_\sR}\psi\neq{\bzero}$ we have
\[ \bZ \cdot \bN = -\frac{\bq_\sR\cdot\nablaR \theta}{\rho_\sR\theta|\partial_{\bq_\sR} \psi|} \]
and then
\beq 
{{\dot\bq}_\sR}=-\left[\frac{\bq_\sR\cdot\nablaR \theta}{\rho_\sR\theta|\partial_{\bq_\sR} \psi|^2} \right] \partial_{\bq_\sR} \psi+(\bone-\bN\otimes\bN)\bJ_\sR\nablaR \theta.
\label{Hypoelastic_mat2} \eeq
Finally, letting
\beq\begin{split}
\bK_{_{R}}(\theta, \bq_\sR)&=\frac{1 }{ \rho_{_R}\theta|\partial_{\bq_\sR} \psi|} \bN\otimes\bq_\sR-(\bone-\bN\otimes\bN)\bJ_\sR\\
&=-\bJ_\sR+\frac{1 }{ \rho_{\sR}|\partial_{\bq_\sR} \psi|^2}\partial_{\bq_\sR} \psi\otimes(\bq_\sR/\theta+\rho_{_R}{\bJ}_\sR^T\partial_{\bq_\sR} \psi),
\end{split}\label{K_R0}\eeq
we can write \eqref{Hypoelastic_mat2} in the more compact form,
\beq
{{\dot\bq}_\sR}=-\bK_{_{R}}\nablaR \theta.
\label{Hypo_mat_2}\eeq
Owing to the arbitrariness of  $\bJ_\sR$ there are infinitely many  tensors $\bK_{_{R}}$   compatible with a given free energy $\psi$. Moreover, $\bK_{_{R}}$  need not  be positive-definite or symmetrical and this is a notable difference from the conductivity tensor in Fourier-like models. 

Thanks to the  aforementioned  correspondence between $ \bT_{\sR\sR},\dot\bE$ and $\bq_\sR,\nablaR (\ln\theta)$,  respectively, an heat conduction theory of this type matches hypoelasticity. Therefore, many achievements from that field can be borrowed here (see \cite{MorroGiorgi_JELAS}).
As an example, let $\psi(\theta,\bq_\sR)$ depend on $\bq_\sR$ via $\xi = |\bq_\sR|^n$, $n \ge 2$. Hence
\[ \partial_{\bq_\sR}\psi = n\, \partial_\xi \psi \,|\bq_\sR|^{n-2}\, \bq_\sR, \quad \partial_\xi \psi\neq0, \]
and  (\ref{K_R0}) becomes
\[\begin{split}
\bK_{_{R}}=-\bJ_\sR+\frac{1 }{ \rho_{\sR} \partial_\xi \psi\, n|\bq_\sR|^{n}}  \, \bq_\sR\otimes\Big(\frac{\bq_\sR}\theta+n\,\rho_{\sR} \partial_\xi \psi \,|\bq_\sR|^{n-2}\, {\bJ}_\sR^T\bq_\sR\Big).
\end{split}\]
Finally, choosing 
$${\bJ}_\sR=-\frac1{n\, \theta\rho_{\sR}\partial_\xi \psi \,|\bq_\sR|^{n-2}}\bone$$  
 we obtain
 $$\bK_{_{R}}=-\bJ_\sR=\frac1{n\, \theta\rho_{\sR}\partial_\xi \psi \,|\bq_\sR|^{n-2}}\bone$$  
and then
$$
{{\dot\bq}_\sR}=-\frac{\nablaR \theta}{n\, \theta\rho_{\sR}\partial_\xi \psi \,|\bq_\sR|^{n-2}}.
$$

On the contrary, if the constitutive relation is given in advance, for instance
\beq
 {{\dot\bq}_\sR}=\hat\bK_{_{R}}\nablaR \theta,
\label{Hypo_mat_hat}\eeq
$\hat\bK_{_{R}}$ being a fixed second-order tensor, one can  ask  whether this model is thermodynamically consistent.
Upon substitution into  \eqref{Nondissipative_cond} we obtain
\[\begin{split}
\rho_{\sR}\partial_{ \bq_\sR}\psi \cdot \hat\bK_{_{R}}\nablaR \theta &=- \frac{\bq_{\sR}}\theta\cdot \nablaR \theta
\end{split}\]
and the arbitrariness of  $\nablaR \theta$ implies
\beq
{\bq_{\sR}}=-\rho_{\sR} \theta\hat\bK^T_{_{R}}\partial_{ \bq_\sR}\psi .
\label{Hypoel_cond0}\eeq
If we replace \eqref{Hypoel_cond0} into \eqref{K_R0} we obtain
\[\begin{split}
\bK_\sR&=\bJ_\sR+\frac{1 }{|\partial_{\bq_\sR} \psi|^2}\partial_{\bq_\sR} \psi\otimes[(\hat\bK^T_{_{R}}-\bJ_\sR^T)\partial_{\bq_\sR}\psi].
\end{split}\]
Then, by  letting  $\bJ_\sR=\hat\bK_{_{R}}$, the identity  $\bK_{_{R}}=\hat\bK_{_{R}}$ follows. We infer that 
\eqref{Hypo_mat_hat} is consistent with the Second Law if and only if there exists a free energy $\psi$ satisfying \eqref{Hypoel_cond0}. This is the case provided that $\hat\bK_{_{R}}$ is a non-singular symmetric constant  tensor and the free energy is given by
\[
\rho_{\sR}\psi=\rho_{\sR}\psi_0(\theta)+\frac1{2\theta}\bq_\sR\cdot\hat\bK^{-1}_\sR\bq_\sR.
\]
 Assuming, for instance,
 $$\hat\bK_\sR=f(\theta, |\bq_\sR|)\bone,$$ 
its thermodynamic consistency is ensured by letting 
 $$\rho_{\sR}\psi=\rho_{\sR}\psi_0(\theta)+\frac1\theta g(\theta,\xi),  \quad\xi = |\bq_\sR|^n, \ n \ge 2,$$
  provided that
$$\partial_\xi g(\theta,\xi) =\frac1{f(\theta,|\bq_\sR|)\,|\bq_\sR|^{n-2}}.$$

\subsection{Maxwell-Cattaneo-Vernotte-like models}
Assume $ \sigma^q > 0$ and $ \partial_{\bq_{\sR}} \psi \neq \bzero$. Equation  \eqref{eq:ineETqq} can then be written in the form
 \[
\frac{\partial_{ \bq_\sR}\psi}{ |\partial_{\bq_\sR} \psi|} \cdot { \dot\bq_\sR}
=- \frac{1 }{ \rho_{_R}|\partial_{\bq_\sR} \psi|}\Big(\frac {\bq_{\sR}} \theta \cdot \nablaR \theta+\rho_\sR\theta\sigma^{q}\Big).
\]
Eexploiting this relation and applying \eqref{eq:BNG} with $\bN=\partial_{\bq_\sR} \psi/|\partial_{\bq_\sR} \psi|$, $\bZ={{\dot\bq}_\sR}$ 
we obtain
\beq
{{\dot\bq}_\sR}=-\Big(\frac {\bq_{\sR}} \theta \cdot \nablaR \theta+\rho_\sR\theta\sigma^{q}\Big) \frac{\bN }{ \rho_{_R}|\partial_{\bq_\sR} \psi|} + (\bone - \bN\otimes \bN)\bG ,
\label{Hypo_mat_CM}\eeq
where $\bG$ is an arbitrary vector-valued function dependent on $(\theta,\bq_\sR,\nablaR \theta)$.

We now show that a class of nonlinear isotropic models of the Maxwell-Cattaneo-Vernotte  (MCV) type \cite{Straughan},
can be derived  from \eqref{Hypo_mat_CM} by properly choosing $\psi$ and $\sigma^q$  (see, for instance, \cite{MorroGiorgi_MDPI}).
Let $\psi(\theta,\bq_\sR)$ depend on $\bq_\sR$ via $\xi = |\bq_\sR|^n$, $n \ge 2$. Hence
\[ \partial_{\bq_\sR}\psi = n\, \partial_\xi \psi \,|\bq_\sR|^{n-2}\, \bq_\sR,  \]
and inequality (\ref{eq:ineETqq}) becomes
\[  \left(n \rho_\sR \partial_\xi \psi |\bq_\sR|^{n-2} \dot{\bq}_\sR + \frac 1 \theta \nablaR \theta\right)\cdot \bq_\sR=-\rho_\sR\theta\sigma^q  \le 0.\]
Let 
\[\sigma^q =\frac{|\bq_\sR|^2}{\rho_\sR\theta^2\kappa},\]
 where $\kappa$ is a positive-valued scalar function. In view of \eqref{Hypo_mat_CM},  if $\partial_\xi \psi\neq0$ it follows
\[ n \rho_\sR \partial_\xi \psi |\bq_\sR|^{n-2} \dot{\bq}_\sR 
= - \frac 1 \theta \nablaR \theta - \frac{1}{\kappa\theta} \bq_\sR .\]
Consequently,
\[\tau \dot{\bq}_\sR 
+ \bq_\sR = - \kappa \nablaR \theta, \qquad \tau =  \kappa n \rho_\sR \theta \partial_\xi \psi |\bq_\sR|^{n-2}, \]
can be viewed as a MCV equation with $\tau$
playing the role of relaxation time and $\kappa$ representing the heat conductivity. If $n=2$ then
\[ \tau = 2\kappa \rho_\sR \theta \partial_\xi \psi . \]
In this case, $\tau$ reduces to a function of the temperature alone provided that $\kappa=\kappa(\theta)$ and $\psi$ is a linear function of $\xi=|\bq_\sR|^2$.

 In addition, a more general class of  anisotropic models of the MCV type can be derived from \eqref{Hypo_mat_CM}. To this end let
\[
\rho_{\sR}\psi=\rho_{\sR}\psi_0(\theta)+\frac\tau{2\theta}\bq_\sR\cdot\bkappa^{-1}\bq_\sR, \qquad \sigma^q =\frac1{\rho_\sR\theta^2}\bq_\sR\cdot\bkappa^{-1}\bq_\sR.
\]
where $\bkappa$ must be a positive-definite second-order tensor in order to have $\sigma^q\ge0$.  In view of  \eqref{Hypo_mat_CM} it follows
\[
{{\dot\bq}_\sR}=-\frac1\tau(\bkappa\bN \cdot \nablaR \theta+\bq_\sR\cdot\bN) {\bN } + (\bone - \bN\otimes \bN)\bG= \bG-\frac1\tau\bN\otimes \bN(\bkappa\nablaR \theta+\bq_\sR+\tau\bG),
\]
where $\bN=\bkappa^{-1}\bq_\sR/|\bkappa^{-1}\bq_\sR|$, and letting $\tau\bG=-\bkappa\nablaR \theta-\bq_\sR$ we obtain 
\beq
\tau{\dot\bq}_\sR+{\bq}_\sR=-\bkappa\nablaR \theta.
\label{eq:MCV}\eeq
The sign of $\tau$ is not prescribed by thermodynamic arguments. However, the common assumption $\tau>0$ implies that $\psi$ has a minimum at ${\bq}_\sR = \bzero$. In this case $\tau$ is called {\it relaxation time} in that we recover the Fourier law when $\tau\to 0$.

\section{Higher-order rate-type local models}
As we are interested in heat conduction only, hereafter we neglect all variables involving stress and strain. 
Accordingly, to describe local effects of higher order in time we expand the previously considered set $(\theta, \bq_\sR, \nablaR \theta)$ by adding first-order time derivatives of these variables.
Hence we let 
\[ \Xi_\sR := (\theta, \dot\theta, \bq_\sR, \dot{\bq}_\sR, \nablaR \theta, \nablaR \dot\theta)\]
be the set of independent and invariant variables. Moreover, let $\psi, \eta, \sigma$ be dependent on $\Xi_\sR$\footnote{Hereafter we neglect the superscript $^q$ in $\sigma^q$.}. Upon evaluation of $\dot{\psi}$ and substitution in (\ref{eq:CD}) we obtain
\[ \begin{split} \rho_{\sR}(\partial_\theta \psi + \eta) \dot{\theta}
+ \rho_{\sR} \partial_{\dot\theta} \psi \cdot \ddot{\theta}  
+ \rho_{\sR}\partial_{ \bq_\sR}\psi \cdot { \dot\bq_\sR}+\rho_{\sR}\partial_{ \dot\bq_\sR}\psi \cdot { \ddot\bq_\sR}& \\
  + \rho_{\sR} \partial_{\nablaR \theta} \psi \cdot \nablaR \dot{\theta}+ \rho_{\sR} \partial_{\nablaR \dot\theta} \psi \cdot \nablaR \ddot{\theta}+\frac 1 \theta \bq_{\sR} \cdot \nablaR \theta&=-\rho_\sR\theta\sigma.\end{split} \] 
The linearity and arbitrariness of $\ddot{\theta}$ and $\nablaR \ddot{\theta}$ imply that $\psi$ is independent of $\dot{\theta}$ and  $\nablaR \dot\theta$,
Assuming, for simplicity, that $\sigma$ is also independent of $\dot\theta$, the linearity and arbitrariness of  $ \dot{\theta}$  imply 
\[ \psi = \psi(\theta,\bq_\sR, \dot\bq_\sR, \nablaR \theta), \qquad \eta=-\partial_\theta \psi , \]
so that $\Sigma_\sR=(\theta,\bq_\sR, \dot\bq_\sR, \nablaR \theta)$ and the entropy inequality reduces to 
\beq \rho_{\sR}\partial_{ \bq_\sR}\psi \cdot { \dot\bq_\sR}+\rho_{\sR}\partial_{ \dot\bq_\sR}\psi \cdot { \ddot\bq_\sR}  + \rho_{\sR} \partial_{\nablaR \theta} \psi \cdot \nablaR \dot{\theta}
+ \frac {\bq_{\sR}} \theta \cdot \nablaR \theta=-\rho_\sR\theta\sigma \le 0. \label{eq:ine1}\eeq 
In the light of the rate-type models considered below,  $\ddot \bq_\sR$ must be regarded as dependent on  $\dot \bq_\sR$,  $\nablaR \theta$ and $\nablaR \dot\theta$.
A similar procedure has been adopted in \cite{MorroGiorgi_MDPI} to evaluate the thermodynamic consistency of some non-linear viscoelastic models of the rate-type, such as the Oldroyd-B and Burgers' fluids.

\subsection{Heat conductors of the Jeffreys type}
The constitutive equation of a heat conductor of the Jeffreys type is given by\footnote{\,$\Sym^+$ denotes the set of  symmetric and positive-definite  tensors of the   second order.}
\beq
\tau\dot \bq_\sR+\bq_\sR=-\bxi\nablaR \theta- \tau\bkappa\nablaR \dot\theta,\qquad  \bxi,\bkappa\in\Sym^+.
\label{eq:Jeffreys}\eeq
where $\tau>0$ is referred to as {\it relaxation time}. This model can also be written as a system
\[
\bq_\sR=-\bkappa\nablaR \theta+\by_\sR, \qquad \tau\dot\by_\sR+ \by_\sR=(\bkappa-\bxi)\nablaR \theta.
\]
 In this way we can interpret the Jeffreys model as the constitutive law of the heat flux vector $\bq_\sR$ in a mixture of two different conductors, $\bq_\sR=\bq_\sR^{(1)}+\bq_\sR^{(2)}$, one of the Fourier type, $\bq_\sR^{(1)}=-\bkappa^{(1)}\nablaR \theta$, $\bkappa^{(1)}=\bkappa$, the other of the MCV type, $ \bq_\sR^{(2)}=\by_\sR$, $\bkappa^{(2)}=\bxi-\bkappa$. A similar result has been obtained in \cite{CR} within the framework of classical irreversible thermodynamics. 

In the limits as $\tau\to 0^+$ and $\tau\to+\infty$ we have, respectively,
\[
\bq_\sR\simeq-\bxi\nablaR \theta, \qquad \dot\bq_\sR\simeq-\bkappa\nablaR \dot\theta.
\]
 Furthermore,  when $\bkappa=\bzero$  the anisotropic MCV model \eqref{eq:MCV} is recovered. 

The rate equation \eqref{eq:Jeffreys} can be derived from  \eqref{eq:ine1} by paralleling the procedure devised in \cite[\S\,8.2]{MorroGiorgi_MDPI} for the Oldroyd-B fluid.
After neglecting the dependence of $\psi$ on $ \dot \bq_\sR$, \eqref{eq:ine1} becomes
\beq
   \rho_{\sR}\partial_{ \bq_\sR}\psi \cdot { \dot\bq_\sR} + \rho_{\sR} \partial_{\nablaR \theta} \psi \cdot \nablaR \dot{\theta}
+ \frac {\bq_{\sR}} \theta \cdot \nablaR \theta=-\rho_\sR\theta\sigma .
\label{eq:heat_rate}\eeq
This equality holds for any pair of functions $\psi$ and $\sigma$, where the former depends on the variables $\bq_\sR,\theta,\nablaR\theta$, whereas the latter depends on the whole set $\Xi_\sR$ and is subject to the thermodynamic condition $\sigma \ge 0$. 

The selection $\bZ=\dot{\bq}_{\sR}$ and $\bN = \partial_{\bq_{\sR}}\psi/|\partial_{\bq_{\sR}}\psi|$ into  the representation formula (\ref{eq:BNG}) together with the exploitation of (\ref{eq:heat_rate}) provide 
\beq \dot{\bq}_{\sR} = -\Big[\frac 1 {\theta}\, \bq_\sR \cdot \nablaR \theta+\rho_\sR \partial_{\nablaR{\theta}} \,\psi\cdot\nablaR \dot{\theta}+ \rho_\sR\theta \sigma\Big]\frac{\partial_{\bq_{\sR}}\psi}{\rho_\sR |\partial_{\bq_{\sR}}\psi|^2}  
+ (\bone- \bN\otimes \bN)\bG,
\label{eq:OG}\eeq
where $\bG$ is an arbitrary second-order tensor-valued function of $\Xi_\sR$. 

Let $\Sym^*$ denotes the set of symmetric and non singular second-order tensors and let
\[ \rho_\sR\psi = \rho_\sR\psi_0(\theta) + 
\frac{\tau}{2\theta } [\bq_{\sR}+ \bkappa\nablaR \theta]\cdot(\bxi+ \bkappa)^{-1} [\bq_{\sR}+ \bkappa\nablaR \theta],\]
with $\bxi+\bkappa\in\Sym^* $. After denoting
\[
\bQ_\sR:= (\bxi+ \bkappa)^{-1}\big(\bq_\sR +\bkappa\nablaR \theta\big),
\]
 we obtain
\[ \rho_\sR \partial_{\bq_\sR} \psi=\frac{\tau}{\theta }\bQ_\sR
 , \qquad \rho_\sR \partial_{\nablaR{\theta}} \,\psi=\frac{\tau}{\theta }\bkappa\bQ_\sR ,\qquad \bN =\frac{\bQ_\sR}{|\bQ_\sR|}.\]
Now we define the entropy production as
\[\rho_\sR \sigma =\frac1{\theta^2}\bq_\sR\cdot(\bxi+ \bkappa)^{-1} \bq_\sR+\frac1{\theta^2}\nablaR \theta\cdot \bkappa(\bxi+ \bkappa)^{-1}\bxi\nablaR \theta.\]
In order to ensure $\sigma\ge0$ we are forced to require 
\begin{itemize}
  \item [i) ] $\bxi\in\Sym^+ $;
  \item [ii) ] there exists $\beta\ge0$ such that $\bkappa=\beta\bxi$.
\end{itemize} 
A straightforward but boring calculation yields
\[
\frac 1 {\theta}\, \bq_\sR \cdot \nablaR \theta+\rho_\sR \partial_{\nablaR{\theta}} \,\psi\cdot\nablaR \dot{\theta}+\rho_\sR\theta \sigma=\frac 1 {\theta}\, [\bq_\sR+\bxi\nablaR \theta+\tau\bkappa\nablaR \dot\theta]\cdot\bQ_\sR;
\]
this result, inserted in equation  \eqref{eq:OG} gives
\[\dot{\bq}_{\sR} = -\frac 1 {\tau}\Big( [\bq_\sR+\bxi\nablaR \theta+\tau\bkappa\nablaR \dot\theta]\cdot\bQ_\sR\Big)
\frac{\bN}{ |\bQ_\sR|}  + (\bone - \bN\otimes \bN)\bG.\]
Taking into account that $(\bN\otimes\bN)\bG=(\bG\cdot\bN)\bN$ and choosing $\tau\bG=-\bq_\sR-\bxi\nablaR \theta-\tau\bkappa\nablaR \dot\theta$ we obtain \eqref{eq:Jeffreys}. This completes the proof of the following Proposition.

\begin{proposition}\label{Prop_3.1}
The thermodynamic consistency of the Jeffreys model is ensured if and only if $\bxi\in\Sym^+$ and $\bkappa=\beta\bxi$, $\beta\ge0$. 
 In the isotropic case\footnote{ A material is said  isotropic if it  has symmetry properties that are invariant with respect to all rotations and inversions of the frame of axes.} $\bxi=\xi\bone$ and $\bkappa=\kappa\bone$, so that the consistency conditions reduce to $\xi>0$, $\kappa\ge0$. 
\end{proposition}

It is worth noting that it is possible to make an alternative choice of  free energy and entropy production. To be precise, we can define
\[ \rho_\sR\psi^*= \rho_\sR\psi_0(\theta) + 
\frac{\tau}{2\theta } [\bq_{\sR}+ \bkappa\nablaR \theta]\cdot(\bxi- \bkappa)^{-1} [\bq_{\sR}+ \bkappa\nablaR \theta],\]
with $\bxi-\bkappa\in\Sym^* $, so that
\[ \rho_\sR \partial_{\bq_\sR} \psi=\frac{\tau}{\theta }\bQ^*_\sR
 , \qquad \rho_\sR \partial_{\nablaR{\theta}} \,\psi=\frac{\tau}{\theta }\bkappa\bQ^*_\sR ,\qquad \bN =\frac{\bQ^*_\sR}{|\bQ^*_\sR|},\]
where
\[
\bQ^*_\sR:= (\bxi-\bkappa)^{-1}\big(\bq_\sR +\bkappa\nablaR \theta\big).
\]
After defining the entropy production as
\[\rho_\sR \sigma^* =\frac1{\theta^2}(\bq_\sR +\bkappa\nablaR \theta\big)\cdot(\bxi-\bkappa)^{-1} (\bq_\sR +\bkappa\nablaR \theta\big)+\frac1{\theta^2}\nablaR \theta\cdot \bkappa\nablaR \theta.\]
we get \eqref{eq:Jeffreys} following the same procedure as above. In this case, however,  $\sigma\ge0$ is ensured by imposing a condition stronger than before, namely $\bxi-\bkappa\in\Sym^+$ and $\bkappa$ positive semidefinite.  In the isotropic case, $\bxi=\xi\bone$ and $\bkappa=\kappa\bone$, this condition reduces to $\xi>\kappa\ge0$. 

\begin{remark}\label{non_uniq}
Each function obtained by the convex combination of ${\psi}$ and ${\psi}^*$ satisfies 
\[
   \rho_{\sR}\partial_{ \bq_\sR}\psi_\lambda \cdot { \dot\bq_\sR} + \rho_{\sR} \partial_{\nablaR \theta} \psi_\lambda \cdot \nablaR \dot{\theta}
+ \frac {\bq_{\sR}} \theta \cdot \nablaR \theta=-\rho_\sR\theta\sigma_\lambda .
\]
where $\psi_\lambda=\lambda{\psi}+(1-\lambda){\psi}^*$ and  $\sigma_\lambda=\lambda{\sigma}+(1-\lambda){\sigma}^*$, $\lambda\in[0,1]$. All of them verify the Clausius-Duhem inequality  provided that $\sigma,\sigma^*\ge0$.
Accordingly,  there is an infinite number of thermodynamically consistent free energy functions for the Jeffreys model. 
\end{remark}

\subsection*{Jeffrey's temperature equation}
To describe the evolution of temperature alone we restrict our attention to a rigid body. Hence, the energy balance equation reads
\beq\rho_\sR c_v\dot\theta = - \nablaR \cdot \bq_\sR + \rho_\sR r.\label{eq:energy}\eeq
After combining  \eqref{eq:Jeffreys} with \eqref{eq:energy} we get
 \begin{equation*}
\rho_\sR c_v(\tau \ddot\theta+\dot\theta)= \nablaR \cdot ( \bxi\nablaR \theta+ \tau \bkappa\nablaR \dot\theta)+ \rho_\sR( r  + \tau \dot r).
\end{equation*}
The temperature evolution of this type of models depends 
 on the relative size of the material parameters.  
For the sake of simplicity, we assume the isotropic case, i.e. $\bxi=\xi\bone$ and $\bkappa=\kappa\bone$ where $\xi,\kappa$ are constant, and 
 let $r=0$ (no external heat source). So we get 
\begin{equation}\label{eqsep}
\tau \ddot\theta+\dot\theta=\frac{1}{\rho_\sR c_v}\left( \xi\nabla_\sR^2\theta+\tau\kappa\nabla_\sR^2\dot\theta\right).
\end{equation}
In \cite{CR1} a similar equation, namely eqn.(45), is derived in the framework od classical irreversible thermodynamics with internal variables. Without reasonable justification, it is therein referred to as Guyer-Krumhansl type temperature equation.

Physically, we assume $\tau>0$ to avoid a singularity at $t=\infty$. Equation \eqref{eqsep} is linear and can be solved by the standard techniques, such as the separation of variables. By setting $\theta(\bX, t)= T(t)Y(\bX)$ we get
\[
\left(\tau\ddot T+\dot T\right)Y=\frac{1}{\rho_\sR c_v}\left(\xi T+\tau \kappa\dot{T}\right)\nabla_\sR^2 Y,
\]
Then applying the separation of variables we obtain
\begin{equation}\label{X_qui}
-\nabla_\sR^2 Y= \Lambda Y,
\end{equation}
\begin{equation}\label{T0}
\tau\ddot T+\dot T=-\tilde\Lambda(\xi T+\tau \kappa\dot{T}), \qquad \tilde\Lambda=\frac{\Lambda}{\rho_\sR c_v}.
\end{equation}
Actually with $\Lambda$ here we indicate a set of constants. For a given domain, $\Lambda$ denotes some eigenvalue of the Laplacian operator as defined by \eqref{X_qui} and the boundary conditions imposed to $Y(\bX)$.  Notice that for the most common boundary conditions (Dirichlet, Neumann, Robin, etc.)~the eigenvalues are countable infinite, let say  $\Lambda_n$, $n=\N$, non-negative and not bounded by any constant value. Usually, the eigenvalues are ordered to form an ascending sequence, $\Lambda_n<\Lambda_{n+1}$.

For any fixed value of $\Lambda_n$, the corresponding solution $T_n$ of \eqref{T0} is given by
\[
T_n(t)=C_1e^{w_+t}+C_2e^{w_{-}t}, \quad w_{\pm}= \frac{-(\tilde\Lambda_n\tau\kappa+1)\pm\sqrt{(\tilde\Lambda_n\tau\kappa+1)^2-4\tilde\Lambda_n\xi\tau}}{2\tau}.
\]
According to Proposition \ref{Prop_3.1}, we let $\kappa\ge0$, $\xi>0$. Consequently, both $w_+$ and $w_{-}$ have negative real part. Indeed, when $\kappa=0$ it follows that $w_{\pm}= (-1\pm\sqrt{1-4\tilde\Lambda_n\xi\tau})/{2\tau}$ are complex with negative real part provided that $n$ is sufficiently large.
For a given small value of $\kappa$ there are, in general, a certain number of eigenvalues $\Lambda_n$ that give rise to complex values of $w_{\pm}$: the number of these eigenvalues increases by decreasing $\kappa$. Depending on the existence of complex values of $w_{\pm}$, the solution $T_n$ can be oscillatory but modulated by a an exponential decay in time. For a numerical characterization of these properties the reader can see e.g. \cite{TN}.   As a last remark, let us notice that the conditions $\kappa\ge 0$ and $\xi > 0$ naturally follow by seeking solutions of \eqref{T0} not blowing up at infinity. Indeed, for any value of $\Lambda_n$, if $T_n(t) = e^{w_nt}$, then $w_n$ solves the quadratic equation
\[
\tau w_n^2+(1+\tau\kappa\tilde\Lambda_n)w+\xi\tilde\Lambda_n=0.
\]
Both solutions of the previous equation has negative real part if and only if all the coefficients have the same sign. Since we let $\tau>0$, we are forced to assume $\xi > 0$ and $1+\tau\kappa\tilde\Lambda_n>0$. The last inequality must hold for all values of $n$, namely
\[
\kappa>-{1}/{\tau\tilde\Lambda_n},
\]
so that taking the limit as $n\to\infty$ we get $\kappa\ge 0$.

\subsection{Green-Naghdi type III heat conductors}
A type III heat conductor according to Green and Naghdi \cite{GN} is characterized by
\[
\bq_\sR=- \bkappa\nablaR\theta-\bxi\nablaR \alpha.
\]
where $\alpha$ is the thermal displacement such that $\dot\alpha=\theta$. As remarked in \cite{BFPG} this relation enters the heat equation for the temperature $\theta$ only after a differentiation with respect to time. Therefore it is reasonable to consider here the rate-type model equation
\beq\dot\bq_\sR=-\bxi\nablaR \theta- \bkappa\nablaR \dot\theta,
\label{eq:GNIII}\eeq
which is obtained from Green-Naghdi type III model by derivation with respect to time when $\bxi,\bkappa$ are constant.

Following the procedure of the previous subsection, to prove the thermodynamic consistency of this model we must provide explicit expressions for $\psi$ and $\sigma$.
We start by letting 
\[ \rho_\sR\psi = \rho_\sR\psi_0(\theta) + 
\frac{1}{2\theta } [\bq_{\sR}+ \bkappa\nablaR \theta]\cdot\bxi^{-1} [\bq_{\sR}+ \bkappa\nablaR \theta],\qquad \bxi\in\Sym^*,\]
and we obtain
\[ \rho_\sR \partial_{\bq_\sR} \psi=\frac{1}{\theta }\bxi^{-1}\big(\bq_\sR +\bkappa\nablaR \theta\big)
 , \quad \rho_\sR \partial_{\nablaR{\theta}} \,\psi=\frac{1}{\theta }\bkappa\bxi^{-1}\big(\bq_\sR +\bkappa\nablaR \theta\big) ,\quad \bN =\frac{\bxi^{-1}\big(\bq_\sR +\bkappa\nablaR \theta\big)}{|\bxi^{-1}\big(\bq_\sR +\bkappa\nablaR \theta\big)|}.\]
Now, if we define the entropy production as
\[\rho_\sR \sigma =\frac1{\theta^2}\nablaR \theta\cdot \bkappa\nablaR \theta,\]
a straightforward calculation yields
\[
\frac 1 {\theta}\, \bq_\sR \cdot \nablaR \theta+\rho_\sR \partial_{\nablaR{\theta}} \,\psi\cdot\nablaR \dot{\theta}+\rho_\sR\theta \sigma=\frac 1 {\theta}\, [\bxi\nablaR \theta+\bkappa\nablaR \dot\theta]\cdot\bxi^{-1}\big(\bq_\sR +\bkappa\nablaR \theta\big).
\]
After replacing this result within the representation formula \eqref{eq:OG} we obtain
\[ \begin{split}
\dot{\bq}_{\sR} = - \big([\bxi\nablaR \theta+\bkappa\nablaR \dot\theta]\cdot\bN\big)
{\bN} 
+ (\bone - \bN\otimes \bN)\bG.
\end{split}
\]
Then choosing the arbitrary vector $\bG=-\bxi\nablaR \theta-\bkappa\nablaR \dot\theta$ we recover \eqref{eq:GNIII}.

\begin{proposition}
The thermodynamic consistency of the GN III model in its differential form \eqref{eq:GNIII} is guaranteed if and only if 
$\bxi\in\Sym^*$ and $\bkappa\in\Sym^+$.
\end{proposition}
It is worth noting that the  GN III model does not proper reduces to the Fourier model, since the solutions of the GN III temperature equation  do not converge to the solutions of the Fourier heat equation when $|\bxi|\to0$ (see \cite{GGP} for a detailed discussion).

\subsection*{GN III temperature equation}
The properties of the temperature evolution in a rigid isotropic GN III heat conductor are described by combining \eqref{eq:energy} and (\ref{eq:GNIII}) to obtain
\begin{equation}\label{eq:GNIII1}
\rho_\sR c_v \ddot\theta= \xi\nabla_\sR^2 \theta+\kappa\nabla_\sR^2 \dot\theta+\rho_\sR \dot r.
\end{equation}
This is a linear version of the GN III temperature equation (see, e.g.,  \cite[eqn.(37)]{BFPG}. 
In absence of heat sources, by setting $\theta(\bX, t)= T(t)Y(\bX)$ and applying the separation of variables we obtain \eqref{X_qui} and
\begin{equation}\label{TGNIII}
\ddot T_n=-\tilde\Lambda_n(\xi T_n+ \kappa\dot{T}_n), \qquad \tilde\Lambda_n=\frac{\Lambda_n}{\rho_\sR c_v},
\end{equation}
where $\Lambda_n$ are the eigenvalues generated by \eqref{X_qui} and a given boundary condition. As $n\in\N$, the values $\Lambda_n$ (and therefore also $\tilde\Lambda_n$) are countable infinite, non-negative, sorted in ascending order and unbounded above.
Due to linearity, we let $T_n(t)=e^{w_nt}$ so that each $w_n$ solves the equation 
\begin{equation}\label{wGNIII}
 w_n^2+\tilde\Lambda_n\kappa w_n+\tilde\Lambda_n\xi =0.
\end{equation}
For all solutions to (\ref{wGNIII}) have negative real part, we must impose that all $\Lambda_n$ are positive and $\kappa>0$, $\xi>0$. For these choices of the parameters the model is thermodynamically consistent and the solutions approaches the steady state exponentially in time. Note that if $\kappa=0$ and $\xi>0$ equation (\ref{wGNIII}) admits purely imaginary roots so that solutions to (\ref{TGNIII}) oscillates over time without decaying. In this case, indeed equation (\ref{eq:GNIII1}) reduces to the wave equation and the GN III model collapses to GN II.

\subsection{Quintanilla's heat conduction model}
A new theory of thermoelasticity phenomena have been proposed by Quintanilla in \cite{Quintanilla} by modifying the Green-Naghdi's type III theory that make use of the thermal displacement $\alpha$ as an independent variable.
 We can obtain the constitutive equation proposed by Quintanilla by adding to the MCV model \eqref{eq:MCV} a term involving the gradient of thermal displacement $\alpha$, 
\[
\tau\dot \bq_\sR+ \bq_\sR=- \bkappa\nablaR\theta-\bxi\nablaR \alpha.
\]
Exploiting the term by term correspondences $ \bT_{\sR\sR}\leftrightarrow \bq_\sR$, $\dot\bE\leftrightarrow-\nablaR \theta$ (and $\bE\leftrightarrow-\nablaR \alpha$) we notice that this equation has a mechanical counterpart in the linear rate-type model of linear viscoelasticity.

Assuming that $\tau,\bxi,\bkappa$ take constant values, deriving with respect to time and taking into account that  $\dot\alpha=\theta$ (the absolute temperature), we obtain the rate-type Quintanilla model
\beq
\tau\ddot \bq_\sR+\dot \bq_\sR=-\bxi\nablaR \theta- \bkappa\nablaR \dot\theta.\label{eq:Q}
\eeq
Equation  \eqref{eq:GNIII} for GN III conductors is recovered when $\tau\to0^+$. 
\begin{proposition}\label{qprop}
The rate-type Quintanilla model \eqref{eq:Q} with $\tau\neq0$ is thermodynamically consistent with the reduced thermodynamic inequality (\ref{eq:ine1}) if and only if $\bxi\in\Sym^*$ and $\bkappa-\tau\bxi\in\Sym^+$. In  the isotropic version
\begin{equation}
\tau\ddot \bq_\sR+\dot \bq_\sR=-\xi\nablaR \theta- \kappa\nablaR \dot\theta,\label{eq:Q_iso}
\end{equation}
that follows from  \eqref{eq:Q} by letting $\bxi=\xi\bone$ and $\bkappa=\kappa\bone$, the consistency condition becomes $\xi\neq0$ and $\kappa>\tau\xi$.
\end{proposition}
\begin{proof}
For simplicity we give the proof in the isotropic case only. 
Let $\tau\neq0$ and consider $\ddot{\bq}_{\sR}$ as a function of   
$ \dot{\bq}_{\sR}, \nablaR \theta , \nablaR \dot\theta$. Upon substitution of $\ddot{\bq}_{\sR}$ from (\ref{eq:Q_iso}) into \eqref{eq:ine1}, we have
\[\begin{split} 
\rho_{\sR}\Big(\partial_{ \bq_\sR}\psi -\frac1\tau\partial_{ \dot\bq_\sR}\psi \Big)\cdot { \dot\bq_\sR}
+ \Big(\frac {\bq_{\sR}} \theta -\frac\xi\tau\rho_{\sR} \partial_{ \dot\bq_\sR}\psi\Big) \cdot\nablaR \theta
& \\
+\rho_{\sR}\Big( \partial_{\nablaR \theta} \psi-\frac{\kappa}\tau\partial_{ \dot\bq_\sR}\psi \Big)\cdot\nablaR \dot\theta 
&=-\rho_\sR\theta\sigma \le 0.\end{split}\]
Since $\psi$ is independent of $\nablaR \dot\theta$, the linearity and arbitrariness of $\nablaR \dot\theta$ imply
\beq \partial_{\nablaR \theta} \psi=\frac{\kappa}\tau\partial_{ \dot\bq_\sR}\psi \label{eq:resQ0}\eeq
and then
 \beq
 \rho_{\sR}\Big(\frac1\tau\partial_{ \dot\bq_\sR}\psi -\partial_{ \bq_\sR}\psi \Big)\cdot { \dot\bq_\sR}
+ \Big(\frac\xi\tau\rho_{\sR} \partial_{ \dot\bq_\sR}\psi-\frac {\bq_{\sR}} \theta \Big) \cdot\nablaR \theta
=\rho_\sR\theta\sigma \ge 0.
\label{eq:redQ}  \eeq
Now we select the dependence on $\bq_\sR,\dot\bq_\sR,\nablaR \theta$ of the free energy $\psi$ in a quadratic form,
\beq\rho_\sR\psi = \rho_\sR\psi_0+\frac{\alpha_1}2 |\bq_{\sR}|^2 +\frac{\alpha_2}2 |\dot\bq_{\sR}|^2+\frac{\alpha_3}2 |\nablaR \theta|^2+ \gamma_1\dot{\bq}_{\sR} \cdot \bq_{\sR}  +\gamma_2\bq_{\sR} \cdot \nablaR \theta +\gamma_3\dot\bq_{\sR} \cdot \nablaR \theta,
\label{eq:psi_Q}\eeq
whence
\[\begin{split}\rho_\sR\partial_{{\bq}_{\sR}}\psi &= \alpha_1 \bq_{\sR}+\gamma_1 \dot\bq_{\sR}+\gamma_2\nablaR \theta, \qquad
\rho_\sR\partial_{\dot\bq_{\sR}}\psi =\gamma_1 \bq_{\sR}+\alpha_2 \dot\bq_{\sR}+\gamma_3\nablaR \theta,  \\
\rho_\sR\partial_{\nablaR \theta}\psi &=\gamma_2 \bq_{\sR}+\gamma_3 \dot\bq_{\sR}+\alpha_3\nablaR \theta.
\end{split}\]
All parameters and $\psi_0$ possibly depend on $\theta$.
Upon substitution into \eqref{eq:resQ0} we obtain
\beq
\gamma_2=  \frac{\kappa}{\tau}\gamma_1, \qquad \gamma_3=  \frac{\kappa}{\tau}\alpha_2,\qquad
\alpha_3 =  \frac{\kappa}{\tau}\gamma_3, 
\label{eq_Q}\eeq
while \eqref{eq:redQ} becomes 
\beq\begin{split}
A_{22} |\dot\bq_{\sR}|^2+ A_{33} |\nablaR \theta|^2+2 A_{12}\dot\bq_{\sR}\cdot\bq_{\sR}
+2 A_{13}\bq_{\sR}\cdot\nablaR \theta+2 A_{23}\dot\bq_{\sR}\cdot\nablaR \theta=\rho_\sR\theta \sigma \ge 0, 
\end{split}\label{ineq_Q}\eeq
where $A_{11}=0$ and 
\[ \begin{split}
 A_{22} &=\frac1\tau\alpha_2-\gamma_1,\quad  A_{33} = \frac{\xi}\tau\gamma_3 ,\quad 2 A_{12}=\frac1\tau\gamma_1-\alpha_1, \\
2 A_{13} &=-\frac1\theta+\frac{\xi}\tau\gamma_1,\quad 
2 A_{23} = \frac1\tau\gamma_3-\gamma_2+ \frac{\xi}\tau\alpha_2.
\end{split}\]

Inequality \eqref{ineq_Q} is satisfied for all values of ${\bq}_{\sR}$, $\dot{\bq}_{\sR}$ and ${\nablaR \theta}$ provided that the symmetric matrix $A=\{A_{ij}\}_{i,j=1,2,3}$  is positive-semidefinite, that is if and only if all principal minors of $A$ are nonnegative (see, for instance \cite[\S\,7.6]{CDM}). Since $A_{11}=0$, we consider the 2-by-2 principal minor
\[d_3:=\det
\begin{pmatrix}
A_{11}    &   A_{12} \\
  A_{12}    &  A_{22}
\end{pmatrix}
=A_{11}A_{22}-A^2_{12}=-\frac1{4\tau^2}(\gamma_1-\tau\alpha_1)^2.\]
If $\gamma_1\neq\tau\alpha_1$ it takes a negative value
and the thesis could not be proved. Hence we let $\alpha_1=\gamma_1/\tau$ so that $A_{12}=0$. Then we consider
\[d_2:=\det
\begin{pmatrix}
A_{11}    &   A_{13} \\
  A_{13}    &  A_{33}
\end{pmatrix}
=A_{11}A_{33}-A^2_{13}=-\frac1{4\theta^2\tau^2}({\xi}\theta\gamma_1-\tau)^2.\]
If $\xi=0$ or $\gamma_1\neq\tau/\xi\theta$ this minor takes a negative value; therefore, we let 
\[\xi\neq0, \qquad \gamma_1=\tau/\xi\theta,\]
 so that $A_{13}=0$. Summarizing and applying \eqref{eq_Q}$_1$
\beq
\gamma_1=\frac\tau{\xi\theta}, \quad \alpha_1=\frac1{\xi\theta}, \quad \gamma_2=\frac\kappa{\xi\theta}, \quad \xi\neq0.
\label{eq:summ}\eeq
Finally we consider
\[\begin{split}d_1:=\det
\begin{pmatrix}
A_{22}    &   A_{23} \\
  A_{23}    &  A_{33}
\end{pmatrix}
&=A_{33}A_{22}-A^2_{23}=\Big(\frac\kappa{\tau^2}\alpha_2-\frac{\kappa}{\xi\theta}\Big) \frac{\xi}{\tau}\alpha_2-\frac1{4}\Big( \frac{\kappa}{\tau^2}\alpha_2-\frac\kappa{\xi\theta}+ \frac{\xi}\tau\alpha_2\Big)^2\\
&=-\frac1{4}\Big( \frac{\kappa}{\tau^2}\alpha_2-\frac\kappa{\xi\theta}- \frac{\xi}\tau\alpha_2\Big)^2
\end{split}\]
where we applied \eqref{eq:summ} and \eqref{eq_Q}$_2$. If the quantity in brackets does not vanish this minor takes a negative value; therefore, we assume it vanishes by letting 
\begin{itemize}
  \item[i) ] either $\kappa=0$ and $\alpha_2=0$,
  \item[ii) ] or $\kappa\neq0$,  $\kappa\neq\tau\xi$ and
\[\alpha_2=\frac{\kappa\tau^2}{\xi\theta(\kappa-\tau\xi)}.\]
\end{itemize}

In the former case (i) we have $\alpha_2=\alpha_3=\gamma_2=\gamma_3=0$ and then
\[A_{11} =A_{33} =A_{12}=A_{13} =A_{23} = 0 ,\quad A_{22} =-\frac\tau{\xi\theta}.
\]
Accordingly, $A$  is positive-semidefinite if and only if $\kappa=0$, $\tau,\xi\neq0$ and $\tau\xi<0$. 

Under the assumptions (ii), exploiting \eqref{eq_Q}$_{2,3}$ we obtain
\[\gamma_3= \frac{\kappa^2\tau}{\xi\theta(\kappa-\tau\xi)},\qquad
\alpha_3 = \frac{\kappa^3}{\xi\theta(\kappa-\tau\xi)},\]
and then we are able to represent $A$ as a function of $\theta$ and the material parameters $\tau,\xi,\kappa$,
\[A_{11}=A_{12}=A_{13}=0, \quad A_{22}=\frac{\tau^2}{\theta(\kappa-\tau\xi)},\quad A_{33}=\frac{\kappa^2}{\theta(\kappa-\tau\xi)},\quad A_{23}=\frac{\kappa\tau}{\theta(\kappa-\tau\xi)}.\]
Accordingly, we infer that $A$  is positive-semidefinite if and only if $\tau,\kappa,\xi\neq0$ and $\kappa>\tau\xi$. 
\end{proof}

It can easily be verified by direct substitution into (\ref{eq:ine1}) that the entropy production related to the Quintanilla model \eqref{eq:Q} takes the form
\[
\rho_\sR \sigma=\frac{1}{\theta^2}(\tau\dot\bq_{\sR}+\bkappa\nablaR \theta)\cdot(\bkappa-\tau\bxi)^{-1} (\tau\dot\bq_{\sR}+\bkappa\nablaR \theta)
\]
and the free energy reads
\[ \begin{split}\rho_\sR\psi = \rho_\sR\psi_0(\theta) &+\frac{1}{2\theta}\big[\tau\dot\bq_{\sR}+\bkappa\nablaR \theta\big]\cdot{\bkappa}{(\bkappa-\tau\bxi)^{-1}} \bxi^{-1} \big[\tau\dot\bq_{\sR}+\bkappa\nablaR \theta\big]
\\&+\frac{1}{2\theta}\Big( \bq_{\sR}
+ 2[\tau\dot{\bq}_{\sR}+\bkappa \nablaR \theta ]\Big)\cdot\bxi^{-1}  \bq_{\sR}.
\end{split}\]
Letting $\tau\to0$ the entropy production and  free energy of the GN  III model are recovered. 
Hence, the Quintanilla model represents a proper generalization of the GN III linear theory, and as such it cannot be considered as
comprehensive of the Fourier theory in a proper sense.

\subsection*{MGT temperature equation}
The properties of temperature evolution can be understood by looking at the equation obtained from \eqref{eq:energy} and (\ref{eq:Q_iso}), namely
\begin{equation}\label{eqqui1}
\tau \dddot\theta + \ddot\theta=\frac{1}{\rho_\sR c_v}\left( \xi\nabla_\sR^2 \theta+\kappa\nabla_\sR^2 \dot\theta\right).
\end{equation}
This looks like a linear version of the well-known Moore-Gibson-Thompson equation which arises in the modeling of wave propagation in viscous thermally relaxing fluids \cite{Thompson}. Since the unknown function of the MGT model is mass density and not temperature, the two equations show a purely formal analogy.

By setting $\theta(\bX, t)= T(t)Y(\bX)$ it follows
\[
\left(\tau\dddot T+\ddot T\right)Y=\frac{1}{\rho_\sR c_v}\left(\xi T+\kappa\dot{T}\right)\nabla_\sR^2 Y.
\]
Then applying the separation of variables we obtain
\begin{equation}\label{Tqui}
\tau\dddot T+\ddot T=-\tilde\Lambda(\xi T+ \kappa\dot{T}),\qquad \tilde\Lambda=\frac{\Lambda}{\rho_\sR c_v},
\end{equation}
where $\Lambda\in\{\Lambda_n\}_{n=\N}$, the sequence of the eigenvalues of the unidimensional Laplacian operator defined by \eqref{X_qui} and the given boundary conditions.
If we let $T(t)=e^{wt}$, $w$ solves the cubic equation 
\begin{equation}\label{wqui}
\tau w^3+ w^2+\tilde\Lambda\kappa w+\tilde\Lambda\xi =0.
\end{equation}
To avoid a blow up of the solutions at infinity, we look at decaying or oscillating solutions of the equation (\ref{Tqui}). It is known that a polynomial $a_3w^3+a_2w^2+a_1w+a_0$ has all three roots with negative real part if and only if all the coefficients are of the same sign and $a_2 a_1>a_0 a_3$ (the so-called Routh-Hurwitz criterion). This provides the following conditions on our parameters,
\begin{equation}\label{quicase}
\tau >0, \; \xi >0, \; \kappa >0, \; \kappa>\tau\xi,
\end{equation}
which are consistent with conditions given in Proposition \ref{qprop}. To get a more precise understanding of the behavior of the solutions, we look at the discriminant of  (\ref{wqui}). After defining
\[
\textrm{Discr$_w$}=\tau^4(w_1-w_2)^2(w_1-w_3)^2(w_2-w_3)^2,
\]
where $w_i$, $i=1,2,3$ are the roots of equation (\ref{wqui}), we get\footnote{\,Given a cubic equation $ax^3 + bx^2 + cx + d = 0$ its discriminant is $ 18abcd - 4b^3d + b^2c^2 - 4ac^3 - 27a^2d^2$.}
\[
\textrm{Discr$_w$}=-\tilde\Lambda\left(4\kappa^3\tau\tilde\Lambda^2+\left(9\tau\xi(3\tau\xi-2\kappa)\right)\tilde\Lambda+4\xi\right).
\]
Accordingly, for any suitable positive value of the constants $\tau$, $\kappa$, $\xi$ and for sufficiently large values of $\Lambda$ the discriminant is negative, and then a couple of roots, say $w_2$ and $w_3$, takes conjugate complex values (with negative real part if the conditions (\ref{quicase}) are fulfilled). Since, in general, the sequence $\{\Lambda_k\}_{k\in\N}$ is unbounded, this implies that solutions to (\ref{Tqui}) oscillates with an exponential damping. This type of behavior is close to that of the MCV model \cite{SZ}. Remarkably,  conditions \eqref{quicase} also ensures well-posedness and stability of solutions in three dimensions \cite{Quintanilla}.


%

\subsection{Heat conductors of the Burgers type}

Heat conductors of the Burgers type are characterized by the rate-type equation
\beq
\lambda\ddot \bq_\sR+\tau\dot \bq_\sR+\bq_\sR=-\bmu\nablaR \theta- \tau\bnu\nablaR \dot\theta.\label{JCL}
\eeq
This model can be obtained, as for the Burgers' rheological model, by considering a mixture of two substances characterized by a conduction mechanims described by the Maxwell-Cattaneo equation (\ref{eq:MCV}), namely
\[
 \tau_1\dot\bq_\sR^{(1)}+ \bq_\sR^{(1)}=-\bkappa^{(1)}\nablaR \theta,
 \qquad
  \tau_2\dot\bq_\sR^{(2)}+ \bq_\sR^{(2)}=-\bkappa^{(2)}\nablaR \theta.
\]
Equation (\ref{JCL}) can be obtained by combining the previous two equations and setting
\beq
\tau=\tau_1+\tau_2,\qquad \lambda=\tau_1\tau_2, \qquad \mu=\bkappa^{(1)}+\bkappa^{(2)}, \qquad \tau\bnu=\tau_1\bkappa^{(2)}+\tau_2\bkappa^{(1)}.
\label{Burg_parameters}\eeq
For simplicity we restrict our attention to the isotropic equation
\beq
\lambda\ddot \bq_\sR+\tau\dot \bq_\sR+\bq_\sR=-\mu\nablaR \theta- \tau\nu\nablaR \dot\theta,\label{eq:BSE_iso}
\eeq
that follows from  \eqref{JCL} by letting $\bmu=\mu\bone$ and $\bnu=\nu\bone$, and we investigate the consistency of this evolution equation with the reduced thermodynamic inequality (\ref{eq:ine1}). 


For definiteness, we let $\lambda\neq0$.  Otherwise, the Burgers model degenerates into simpler models. The Jeffreys model is recovered if $\lambda=0$. When $\lambda$ and $\bnu$ vanish, then \eqref{eq:BSE_iso} reduces to the MCV model and if in addition $\tau=0$ the Fourier law is recovered.
Consequently, the Burgers' model, unlike the GN III and Quintanilla's models, contains the Fourier theory in a proper sense.

Owing to $\lambda\neq0$, we consider $\ddot{\bq}_{\sR}$ as a function of   
${\bq}_{\sR}, \dot{\bq}_{\sR}, \nablaR \theta , \nablaR \dot\theta$. Upon substitution of $\ddot{\bq}_{\sR}$ from (\ref{eq:BSE_iso}) into \eqref{eq:ine1}, we have
\[\begin{split} 
\rho_{\sR}\Big(\partial_{ \bq_\sR}\psi -\frac\tau\lambda\partial_{ \dot\bq_\sR}\psi \Big)\cdot { \dot\bq_\sR}
-\rho_{\sR}\frac1\lambda\partial_{ \dot\bq_\sR}\psi \cdot  \bq_\sR
+ \Big(\frac {\bq_{\sR}} \theta -\frac\mu\lambda\rho_{\sR} \partial_{ \dot\bq_\sR}\psi\Big) \cdot\nablaR \theta
& \\
+\rho_{\sR}\Big( \partial_{\nablaR \theta} \psi-\frac{\tau\nu}\lambda\partial_{ \dot\bq_\sR}\psi \Big)\cdot\nablaR \dot\theta 
&=-\rho_\sR\theta\sigma \le 0.\end{split}\]
Since $\psi$ is independent of $\nablaR \dot\theta$, the linearity and arbitrariness of $\nablaR \dot\theta$ imply
\beq \partial_{\nablaR \theta} \psi=\frac{\tau\nu}\lambda\partial_{ \dot\bq_\sR}\psi \label{eq:resBur0}\eeq
and then
 \beq
 \rho_{\sR}\Big(\partial_{ \bq_\sR}\psi -\frac\tau\lambda\partial_{ \dot\bq_\sR}\psi \Big)\cdot { \dot\bq_\sR}
-\rho_{\sR}\frac1\lambda\partial_{ \dot\bq_\sR}\psi \cdot  \bq_\sR
+ \Big(\frac {\bq_{\sR}} \theta -\frac\mu\lambda\rho_{\sR} \partial_{ \dot\bq_\sR}\psi\Big) \cdot\nablaR \theta
=-\rho_\sR\theta\sigma \le 0.
\label{eq:redBur}  \eeq
Now we select  the free energy $\psi$ as in \eqref{eq:psi_Q},  so that
upon substitution into \eqref{eq:resBur0} and \eqref{eq:redBur} we obtain
\beq
\gamma_2=\frac{\tau\nu}\lambda\gamma_1,\qquad \gamma_3 =\frac{\tau\nu}\lambda\alpha_2,\qquad \alpha_3 =\frac{\tau\nu}\lambda\gamma_3,
\label{eq:eq}\eeq
\beq\begin{split}
A_{11} |\bq_{\sR}|^2+A_{22} |\dot\bq_{\sR}|^2+ A_{33} |\nablaR \theta|^2+2 A_{12}\dot\bq_{\sR}\cdot\bq_{\sR}&
\\+2 A_{13}\bq_{\sR}\cdot\nablaR \theta+2 A_{23}\dot\bq_{\sR}\cdot\nablaR \theta&=\rho_\sR\theta \sigma \ge 0, 
\end{split}\label{ineq}\eeq
where
\[ \begin{split}
A_{11} &=\frac1\lambda\gamma_1,\quad A_{22} =\frac\tau\lambda\alpha_2-\gamma_1,\quad A_{33} = \frac{\mu}\lambda\gamma_3 ,\\
2A_{12}&=\frac\tau\lambda\gamma_1+\frac1\lambda\alpha_2-\alpha_1,\quad 2A_{13} =-\frac1\theta+\frac1\lambda
\gamma_3+\frac{\mu}\lambda\gamma_1,\quad 2A_{23} = \frac\tau\lambda\gamma_3-\gamma_2+ \frac{\mu}\lambda\alpha_2.
\end{split}\]

\begin{proposition}\label{Burcon}
The Burgers-like model \eqref{eq:BSE_iso} with $\lambda\neq0$ is thermodynamically consistent if and only if one of the following hypotheses occurs 
\begin{itemize}
  \item[i)] $\tau\nu = 0$ and $\mu>0$, $\lambda<0$;
  \item[ii)]  $\nu\tau\neq0$ and $\mu=0$, $\nu>0$;
   \item[iii)]  $\nu\tau\neq0$ and $\mu>0$, $\nu\tau^2\ge\lambda\mu$.
\end{itemize}
\end{proposition}

\begin{proof}
To satisfy inequality \eqref{ineq} for all values of ${\bq}_{\sR}$, $\dot{\bq}_{\sR}$ and ${\nablaR \theta}$ we have to prove that the symmetric matrix $A=\{A_{ij}\}_{i,j=1,2,3}$  is positive semidefinite. First we consider the 2-by-2 principal minor
\[\begin{split}d_1:=\det
\begin{pmatrix}
A_{22}    &   A_{23} \\
  A_{23}    &  A_{33}
\end{pmatrix}
&=A_{22}A_{33}-A^2_{23}=\Big(\frac\tau{\lambda}\alpha_2-\gamma_1\Big) \frac{\mu}{\lambda}\gamma_3-\frac1{4}\Big(\frac\tau\lambda\gamma_3-\gamma_2+ \frac{\mu}\lambda\alpha_2\Big)^2\\
&=-\frac1{4\lambda^2}\Big( \frac{\nu\tau^2-\mu\lambda}{\lambda}\alpha_2-{\nu\tau}\gamma_1\Big)^2
\end{split}\]
where we applied \eqref{eq:eq}$_{1,2}$. If the quantity in brackets does not vanish this minor takes a 
negative value; therefore, we are forced to impose $d_1=0$. This can be achieved in different ways depending on the parameter values: 
\begin{enumerate}
  \item when $\tau\nu = 0$ and $\mu = 0$, whatever the values of $\gamma_1$ and $\alpha_2$ are.
  \item when $\tau\nu = 0$ and $\mu \neq 0$,  by setting $\alpha_2=0$.
  \item when $\nu\tau\neq0$ and $\mu=0$ by setting 
  \beq\gamma_1=\frac{\tau}{\lambda}\alpha_2.\label{eq:alpha_2_0}\eeq
  \item when $\nu\tau\neq0$ and  $\mu\neq0$ by setting 
  \beq\gamma_1=\frac{\nu\tau^2-\mu\lambda}{\nu\tau\lambda}\alpha_2.\label{eq:alpha_2}\eeq
\end{enumerate}

\smallskip
{\it Case} (1). If $\mu = 0$ and either $\tau=0$ or $\nu= 0$ then  from \eqref{eq:eq} it follows $\gamma_2=\gamma_3=\alpha_3=0$ and 
\[ 
A_{11} =\frac1\lambda\gamma_1,\quad A_{33} = 0 ,\quad A_{13} =-\frac1{2\theta},\quad A_{23} = 0.
\]
Hence, in any case $d_2:=A_{11}A_{33}-A^2_{13}=-1/4\theta^2<0$ and $A$ cannot be positive semidefinite.

\smallskip
{\it Case} (2). If $\mu \neq 0$ and either $\tau=0$ or $\nu= 0$ then  from \eqref{eq:eq} it follows 
$\gamma_2=\gamma_3=\alpha_3=0$. By setting $\alpha_2=0$ we get $d_1=0$ and
\[ \begin{split}
A_{11} &=\frac1\lambda\gamma_1,\quad A_{22} =-\gamma_1,\quad A_{33} = 0 ,\\
2A_{12}&=\frac\tau\lambda\gamma_1-\alpha_1,\quad 2A_{13} =-\frac1\theta+\frac{\mu}\lambda\gamma_1,\quad 2A_{23} = 0.
\end{split}\]
Since 
\[d_2=-\frac14\Big(-\frac1\theta+\frac{\mu}\lambda\gamma_1\Big)^2\]
we are forced to set $\gamma_1=\frac\lambda{\mu\theta}$. Accordingly, $d_2=0$ and
\[A_{11} =\frac1{\mu\theta},\qquad A_{22} =-\frac\lambda{\mu\theta},\qquad 2A_{12}=\frac\tau{\mu\theta}-\alpha_1,\]
so that the diagonal elements of $A$ are non negative only if $\mu>0$ and $\lambda<0$. If this is the case,
\[d_3:=A_{11}A_{22}-A^2_{12}=\frac{|\lambda|}{\mu^2\theta^2}-\frac14\Big(\frac\tau{\mu\theta}-\alpha_1\Big)^2\ge0\]
is ensured to be positive by choosing, for instance, $\alpha_1=\tau/{\mu\theta}$ so that $A_{12}=0$. Finally, 
\[\det A=d_1A_{11} -d_2A_{22} +d_3A_{33}=0\]
and we conclude that $A$ is positive semidefinite provided that i) occurs.

\smallskip
{\it Case} (3). Let $\nu\tau\neq0$ and $\mu=0$. After replacing \eqref{eq:alpha_2_0} into \eqref{eq:eq} we obtain
\[\gamma_2=\frac{\nu\tau^2}{\lambda^2}\alpha_2,\qquad \gamma_3 =\frac{\tau\nu}\lambda\alpha_2,\qquad \alpha_3 =\frac{\tau^2\nu^2}{\lambda^2}\alpha_2,\]
and then $A_{22} =A_{33} = A_{23} = 0$
\[ \begin{split}
A_{11} =\frac{\tau}{\lambda^2}\alpha_2,\qquad 
2A_{12}=\frac{\tau^2+\lambda}{\lambda^2}\alpha_2-\alpha_1,\qquad
 2A_{13} =-\frac1\theta+
\frac{\nu\tau}{\lambda^2}\alpha_2.
\end{split}\]
By letting
\[\alpha_1=\frac{\tau^2+\lambda}{\lambda^2}\alpha_2, \qquad  \alpha_2=\frac{\lambda^2}{\nu\tau\theta}\]
we get $d_2=d_3=0$ and $A$ turns out to be positive semidefinite if and only if $A_{11} = 1/\nu\theta\ge0$, that is when  ii) occurs.

\smallskip
{\it Case} (4). Let $\nu\tau\neq0$, $\mu\neq0$. After replacing \eqref{eq:alpha_2} into \eqref{eq:eq} we obtain
\[\gamma_2=\frac{\nu\tau^2-\mu\lambda}{\lambda^2}\alpha_2,\qquad \gamma_3 =\frac{\tau\nu}\lambda\alpha_2,\qquad \alpha_3 =\frac{\tau^2\nu^2}{\lambda^2}\alpha_2,\]
and then
\[ \begin{split}
A_{11} &=\frac{\nu\tau^2-\mu\lambda}{\nu\tau\lambda^2}\alpha_2,\quad A_{22} =\frac{\mu}{\nu\tau}\alpha_2,\quad A_{33} = \frac{\mu\nu\tau}{\lambda^2}\alpha_2 ,\quad 
2A_{12}=\frac{\nu\tau^2+(\nu-\mu)\lambda}{\nu\lambda^2}\alpha_2-\alpha_1,\\
 2A_{13} &=-\frac1\theta+
\frac{\nu^2\tau^2+\mu(\nu\tau^2-\mu\lambda)}{\nu\tau\lambda^2}\alpha_2,\quad 2A_{23} = \frac{2\mu}\lambda\alpha_2.
\end{split}\]
If we assume $A_{12} =0$ by letting
\[
\alpha_1=\frac{\nu\tau^2+(\nu-\mu)\lambda}{\nu\lambda^2}\alpha_2
\]
then $d_3=A_{11}A_{22}$ and since $d_1=0$, $d_2=A_{11}A_{33}-A_{13}^2$ it follows 
$$\det A=d_3A_{33}-d_2A_{22}=(A_{11}A_{33}-d_2)A_{22}=A_{13}^2A_{22}.$$
If we choose 
\[\alpha_2=\frac{\nu\tau\lambda^2}{\theta[\nu^2\tau^2+\mu(\nu\tau^2-\mu\lambda)]}\]
then $A_{13}=0$ and $A$ is positive semidefinite if and only if its diagonal terms 
\[A_{11} =\frac{\nu\tau^2-\mu\lambda}{\theta[\nu^2\tau^2+\mu(\nu\tau^2-\mu\lambda)]},\quad 
A_{22} =\frac{\mu\lambda^2}{\theta[\nu^2\tau^2+\mu(\nu\tau^2-\mu\lambda)]},\quad 
A_{33} =\frac{\mu\nu^2\tau^2}{\theta[\nu^2\tau^2+\mu(\nu\tau^2-\mu\lambda)]}\]
 take nonnegative values.
This is achieved by assuming conditions iii) on the material parameters.
\end{proof}

\subsection*{Joseph-Preziosi temperature equation}
Let us briefly discuss the temperature equation corresponding to Burgers-like model. For simplicity we restrict the discussion to the rigid and  isotropic case with all the parameters constant. Also, we assume there is no external heat supply. By combining the energy equation \eqref{eq:energy}
with \eqref{eq:BSE_iso} we get
\beq
\lambda \dddot\theta + \tau \ddot\theta+\dot\theta=\frac{1}{\rho_\sR c_v}\left( \mu\nablaR^2 \theta+\tau\nu\nablaR^2 \dot\theta\right).
\label{JP_temp_eq}\eeq
This third order equation looks like a generalization of the MGT temperature equation. Really, an equation of this form (see \cite[eqn.(5.7)]{JP}) had already been obtained by Joseph and Preziosi starting from the linearization of heat conduction with memory according to the Gurtin-Pipkin model \cite{GP}. 

Letting $\theta(\bx, t)= T(t)Y(\bX)$, we obtain the temperature equation
\[
\left(\lambda \dddot T+\tau\ddot T+\dot T\right)Y=\frac{1}{\rho_\sR c_v}\left(\mu T+\tau \nu\dot{T}\right)\nabla_\sR^2 Y,
\]
and then the equation for $T(t)$ reads
\begin{equation}\label{T}
\lambda \dddot T + \tau\ddot T+\dot T=-\tilde\Lambda(\mu T+\tau \nu\dot{T}),\qquad \tilde\Lambda=\frac{\Lambda}{\rho_\sR c_v},
\end{equation}
where  $\tilde\Lambda$ is the set of all eigenvalues $\tilde\Lambda_n$, $n\in\N$, that we remember must be non-negative and unbounded in order to comply with the given boundary conditions.
If $T(t)=e^{wt}$, then $w$ solves the cubic equation 
\begin{equation}\label{wpo}
\lambda w^3+\tau w^2+(1+\tilde\Lambda_n\tau\nu)w+\tilde\Lambda_n\mu =0, \qquad n\in\N.
\end{equation}
To avoid blow up of solutions at large time, we seek only for decaying or oscillating solutions to equation (\ref{T}). Accordingly, we look at the conditions giving all roots with negative real part. By applying the Routh-Hurwitz criterion, for any $n\in\N$ we get
\begin{equation}\label{RHu}
\lambda >0, \;\tau >0, \ 1+\tilde\Lambda_n\tau\nu>0,\ \tilde\Lambda_n\mu>0,\ \tau+\tilde\Lambda_n(\tau^2\nu-\lambda\mu)>0.
\end{equation}
or 
\begin{equation}\label{RHu1}
\lambda <0, \ \tau <0, \ 1+\tilde\Lambda_n\tau\nu<0,\ \tilde\Lambda_n\mu<0,\ \tau+\tilde\Lambda_n(\tau^2\nu-\lambda\mu)<0.
\end{equation}
Let us look at the compatibility of (\ref{RHu}) or (\ref{RHu1}) with the conditions about the thermodynamical consistency stated in Proposition \ref{Burcon}. 

First consider the compatibility of condition i),  $\mu>0$, $\lambda<0$ and $\tau\nu=0$, with \eqref{RHu1}. In this case  the polynomial (\ref{wpo}) takes the form
\[
\lambda w^3+\tau w^2+w+\tilde\Lambda_n\mu.
\]
and \eqref{RHu1} is violated in that the coefficients of $w^3$ and $w$ have different sign. We conclude that condition i) is thermodynamically consistent but gives an unphysical  model from the point of view of the behavior of the solution, since it diverges for large times. 

Then assume condition ii), $\tau\neq0$, $\mu=0$, $\nu>0$. In this case a root of (\ref{wpo}) is equal to zero, the other two roots being solution of 
\[
\lambda w^2+\tau w+(1+\tilde\Lambda_n\tau\nu) =0.
\]
The coefficients of the above polynomial must all have the same sign for both of its roots to have a negative real part. This is achieved by assuming either $\lambda >0$, $\tau >0$ or 
\beq\lambda <0,\quad \tau <0, \quad1-\tilde\Lambda_n|\tau|\nu<0.\label{Bur_strange}\eeq 
The last inequality must be satisfied for all the eigenvalues $\Lambda_n$. Taking into account that  $\Lambda _n$ is an ascending sequence and $\nu >0$, this is only true if
\[
\Lambda_1>0, \qquad\nu>\frac{\rho_\sR c_v}{\Lambda_1 |\tau|}.
\]
Such a relation is very peculiar since the actual value of $\Lambda_1$ depends on the given boundary conditions and the size of the domain. 

Finally, let us consider the case iii), i.e $\nu\tau\neq0$, $\mu>0$ and $\nu\tau^2\ge\lambda\mu$. In this case inequalities (\ref{RHu}) are compatible with iii)  and the corresponding physical model is thermodinamically consistent with solutions decaying as $t\to\infty$. On the contrary, conditions (\ref{RHu1}) conflict with iii) and must be rejected. Indeed, if $\tau<0$, then inequalities $\tau^2\nu-\lambda\mu>0$ of iii) and $\tau+\tilde\Lambda_n(\tau^2\nu-\lambda\mu)<0$ of (\ref{RHu1}) would give for any $n\in\N$
$$\tau<-\tilde\Lambda_n(\tau^2\nu-\lambda\mu).$$
which leads to an unbounded value of $\tau$ since $\Lambda_n$ belongs to an unbounded sequence. 

Summarizing: the Burgers-like model are thermodinamically and dinamically consistent (i.e. the entropy production is nonnegative and all solutions of the temperature equation are bounded and eventually approach equilibrium values) if and only if the following conditions are satisfied:
\[
\mu \geq 0,\; \nu>0,\; \lambda>0, \; \tau>0, \; \nu\tau^2\ge\lambda\mu.
\]
We excluded the case \eqref{Bur_strange} due to its peculiarity. This exclusion is also suggested by the observation that the positivity of $\lambda$ and $\tau$ could be also inferred from the physical assumptions we made to construct the Burgers-like models. Indeed, by \eqref{Burg_parameters} $\lambda$ and $\tau$ are respectively the product and the sum of two positive relaxation times.

\subsection{Local heat conduction models of rate-type in spatial description}
The models described in the previous sections can be applied also to fluid heat conductors by rewriting their constitutive rate-type equations into the spatial description.
We illustrate the procedure by considering the Burgers-like model \eqref{JCL}
\[\lambda\ddot \bq_\sR+\tau\dot \bq_\sR+\bq_\sR=-\bmu\nablaR \theta- \tau\bnu\nablaR \dot\theta.\]
To this end we remember that $\nablaR \dot\theta=(\nablaR \theta)^\cdot$ and
$$\bq_\sR = J \bF^{-1} \bq,\qquad \nablaR=\bF^{T}\nabla.$$
Therefore it follows
\[\lambda (J \bF^{-1} \bq)^{\cdot\cdot}+\tau(J \bF^{-1} \bq)^{\cdot}+J \bF^{-1} \bq=-\bmu\bF^{T}\nabla \theta- \tau\bnu(\bF^{T}\nabla \theta)^\cdot.\]
After introducing the  convective time derivative (also referred to as {\it Truesdell vector rate}  \cite[p.59]{GM_book}),
$$\quadro{\bq}\,=\partial_t\bq+\nabla\times(\bq\times\bv)+(\nabla\cdot\bq)\bv=\dot\bq-\bL\bq+(\nabla\cdot\bv)\bq ,$$
and the Cotter-Rivlin vector rate,
$$(\nabla\theta)^\strianup= (\nabla\theta)^\cdot+\bL^T\nabla\theta,$$
a straightforward calculation provides
\[(J \bF^{-1} \bq)^{\cdot}=J \bF^{-1}\quadro\bq, \qquad (J \bF^{-1} \bq)^{\cdot\cdot}=J \bF^{-1}\quadri\bq, \qquad (\bF^{T}\nabla \theta)^\cdot=\bF^{T}(\nabla\theta)^\strianup.\]
Accordingly, we can write
\[J \bF^{-1} [\lambda\quadri\bq+\tau\quadro\bq+\bq]=-\bmu\bF^{T}\nabla \theta- \tau\bnu\bF^{T}(\nabla\theta)^\strianup,\]
from which it follows
\beq\lambda\quadri\bq+\tau\quadro\bq+\bq=-\bmu_s\nabla \theta- \tau \bnu_s(\nabla\theta)^\strianup,\label{Burgers_spatial}\eeq
where 
$$\bmu_s:=J^{-1} \bF\bmu\bF^{T}, \qquad\bnu_s:=J^{-1} \bF\bnu\bF^{T}$$
are the viscosity tensors in the spatial description. In particular,  when $\bmu$ and $\bnu$ are isotropic, namely $\bmu=\mu\bone$ and $\bnu=\nu\bone$, then
$$\bmu_s=J ^{-1}\mu\bB, \qquad\bnu_s=J^{-1}\nu\bB, \qquad \bB:=\bF\bF^{T}.$$
Neglecting that the viscosity tensors depend on $\bF$, \eqref{Burgers_spatial} represents a linear rate-type constitutive equation for the heat flux vector in the spatial description. Note that all time derivatives involved there are objective.

\section{Weakly nonlocal theories of heat conduction}
Nonlocal effects in generalized equations for heat conduction are of particular interest in systems where the mean free path is comparable to (or bigger than) the characteristic volume size, i.e., where the Knudsen number is equal or greater than one. This may occur when considering nanosystems (the  size of the system is comparable to the mean free path) or studying propagation of phonons at low temperatures or in rarefied systems (the mean free path is comparable to the size of the system).

We focus here on {\it non-local} theories by assuming  $\sigma=\zeta+\nablaR\cdot\bk_\sR$ subject to $\zeta\ge0$ and $\bk_\sR\cdot\bn_\sR\vert_{\partial\cR}=0$, in agreement with the nonlocal statement of the Second Law \eqref{strong_II_law}. 
As in the previous sections, we neglect the dependence on ${\bE}$ and $ {\bT}_{\sR\sR}$.  Hence \eqref{basic_ent_ineq} reduces to 
\beq-{\rho_\sR}(\dot{\psi} + \eta \dot{\theta}) 
- \frac 1 {\theta}\, \bq_\sR \cdot \nablaR \theta+\theta\nablaR\cdot\bk_\sR=\rho_\sR\theta\zeta  \ge 0.\label{eq:CD_nl} \eeq

Weakly nonlocal models  are developed here on the basis of four elements,
\begin{itemize}
\item[-] the set $\Sigma_\sR$ of admissible variables;
\item[-] the free energy density function $\psi=\psi(\Sigma_\sR)$;
\item[-] the internal entropy production function $\zeta=\zeta(\Sigma_\sR)$
\item[-] the extra entropy flux function $\bk=\bk(\Sigma_\sR)$
\end{itemize}

\subsection{Basic nonlocal models}
To describe weakly nonlocal effects in heat conduction,  we add $\nablaR\bq_\sR$ to the set of variables.
Hence,
\[ \Sigma_\sR := (\theta, \bq_\sR, \nablaR \theta, \nablaR\bq_\sR).\]
is  the basic set of independent variables and  $\psi, \eta$, $\sigma$ and $\bk$ are assumed to 
depend on $\Sigma_\sR$. In addition, we assume $\eta$ is continuous while $\psi$ is continuously differentiable.
Upon evaluation of $\dot{\psi}$ and substitution in (\ref{eq:CD_nl}) we obtain
\[ \begin{split} \rho_{\sR}(\partial_\theta \psi + \eta) \dot{\theta}
+ \rho_{\sR} \partial_{ \bq_\sR}\psi \cdot { \dot\bq_\sR}  + \rho_{\sR} \partial_{\nablaR \bq} \psi \cdot \nablaR \dot\bq_\sR& \\
 + \rho_{\sR} \partial_{\nablaR \theta} \psi \cdot \nablaR \dot{\theta}+\frac 1 \theta \bq_{\sR} \cdot \nablaR \theta+\theta\nablaR\cdot\bk_\sR&=-\rho_\sR\theta\zeta\le0.\end{split} \] 
The linearity and arbitrariness of $\dot{\theta}$, $\dot\bq_\sR$, $\nablaR \dot\bq_\sR$ and $\nablaR \dot{\theta}$ imply that 
\[ \psi = \psi(\theta), \qquad \eta = - \partial_\theta \psi, \]
and (\ref{eq:CD_nl}) reduces to 
\beq  \frac {\bq_{\sR}} {\theta^2} \cdot \nablaR \theta+\nablaR\cdot\bk_\sR=-\rho_\sR\zeta \le 0. \label{eq:ine_nl}\eeq 

For definiteness,  we exhibit a nonlocal perturbation of the Fourier law which is consistent with the Second Law in the form \eqref{eq:ine_nl}. According to \cite{GK2} we let the following macroscopic constitutive equation,
\beq
\bq_\sR=-\kappa \nablaR\theta+\lambda^2[\nabla_\sR^2\bq_\sR+2\nablaR(\nablaR\cdot\bq_\sR)],
\label{eq:nlFourier}\eeq
where $\kappa(\theta)$ is the usual bulk thermal conductivity and $\lambda^2(\theta)$ is a (possibly small) positive function. To prove its thermodynamic consistency  we  let $\psi=\psi(\theta)$ and
\beq
\rho_\sR\zeta =\frac{\ell^2}{\lambda^2} |\bq_\sR|^2+\ell^2\big[|\nablaR\bq_\sR|^2+2|\nablaR\cdot\bq_\sR|^2\big], \quad \bk_\sR=-\ell^2\big[\bq_\sR\nablaR\bq_\sR+2(\nablaR\cdot\bq_\sR)\bq_\sR\big].
\label{eq:nl_F_diss} \eeq
where $\ell$ is an undetermined constant parameter. After computing
\[
\nablaR\cdot\bk_\sR=-\ell^2\Big[|\nablaR\bq_\sR|^2+2|\nablaR\cdot\bq_\sR|^2+[\nablaR\cdot\nablaR\bq_\sR+2\nablaR(\nablaR\cdot\bq_\sR)]\cdot\bq_\sR\Big]
\]
and noticing that $\nablaR\cdot\nablaR=\nabla_\sR^2$, from  \eqref{eq:ine_nl} and \eqref{eq:nl_F_diss} it follows
\beq \frac {\bq_{\sR}} {\theta^2} \cdot \nablaR \theta=-\frac{\ell^2}{\lambda^2} |\bq_\sR|^2+{\ell^2}\big[\nabla_\sR^2\bq_\sR+2\nablaR(\nablaR\cdot\bq_\sR)\big]\cdot\bq_\sR.\label{eq:intemedia}\eeq
So, after taking $\nablaR\theta$ from equation \eqref{eq:nlFourier} and substituting it into \eqref{eq:intemedia} we obtain
 \[ \Big[\frac {1} {\kappa\theta^2}-\frac{\ell^2}{\lambda^2}\Big] \bq_{\sR}\cdot \Big[-\bq_\sR+\lambda^2[\nabla_\sR^2\bq_\sR+2\nablaR(\nablaR\cdot\bq_\sR)]\Big]=0.\]
This means that \eqref{eq:nlFourier} is thermodynamically consistent provided that $\kappa$ and $\lambda$ are related by
\[\kappa(\theta)=\frac {\lambda^2(\theta)} {\ell^2\theta^2}.\]
Since $\ell$ is an arbitrary constant parameter, it is then sufficient that $\kappa$ is positive and proportional to $\lambda^2/\theta^2$.

\subsection{Rate-type nonlocal models}
In this case the set of variables is allowed to contain the time derivative of the heat flux vector,
\[ \Sigma_\sR := (\theta,  \bq_\sR, \nablaR \theta, \dot{\bq}_\sR,\nablaR\bq_\sR)\]
and constitutive functions for $\psi, \bk$ and $\zeta$ depend on $\Sigma_\sR$.
Upon evaluation of $\dot{\psi}$ and substitution in (\ref{eq:CD}) we obtain
\[ \begin{split} \rho_{\sR}(\partial_\theta \psi + \eta) \dot{\theta}
+ \rho_{\sR} \partial_{ \bq_\sR}\psi \cdot { \dot\bq_\sR}+ \rho_{\sR} \partial_{\nablaR \theta} \psi \cdot \nablaR \dot{\theta} + \rho_{\sR} \partial_{ \dot\bq_\sR}\psi \cdot { \ddot\bq_\sR}   & \\
+ \rho_{\sR} \partial_{\nablaR \bq} \psi \cdot \nablaR \dot\bq_\sR
 +\frac 1 \theta \bq_{\sR} \cdot \nablaR \theta+\nabla\cdot\bk_\sR&=-\rho_\sR\zeta.\end{split} \] 
The linearity and arbitrariness of $\dot{\theta}$, $\ddot{\bq}_\sR$, $\nablaR \dot{\theta}$ and $\nablaR \dot\bq_\sR$ imply that 
\[ \psi = \psi(\theta,\bq_\sR), \qquad \eta=-\partial_\theta \psi, \]
so that the Clausius-Duham inequality reduces to 
\beq   \rho_{\sR}\partial_{ \bq_\sR}\psi \cdot { \dot\bq_\sR}
+ \frac {\bq_{\sR}} \theta \cdot \nablaR \theta+\nabla\cdot\bk_\sR=-\rho_\sR\zeta \le 0. \label{eq:ine0_nl}\eeq

\subsubsection{Guyer-Krumhansl conductors}
Guyer and Krumhansl \cite{GK} studied the heat wave propagation in dielectric crystals at low temperature. They observed that in the regime of low temperature the heat flux $\bq$ is proportional to the momentum flux of the phonon gas and then found the following macroscopic equation governing its evolution 
\beq
\tau\dot \bq_\sR+\bq_\sR+\kappa \nablaR\theta=\lambda^2[\nabla_\sR^2\bq_\sR+2\nablaR(\nablaR\cdot\bq_\sR)],
\label{eq:GrK}\eeq
where $\tau$ is the relaxation time for resistive phonon scattering, $\lambda^2$ is the mean-free path of phonons and $\kappa$ is the usual bulk thermal conductivity \cite{AGr}. Both $\lambda$ and $\kappa$ possibly depend on $\theta$.
Hereafter we prove the consistency of this linear equation with the Second Law in the nonlocal form \eqref{eq:ine0_nl}.

Assuming $ \partial_{\bq_{\sR}} \psi \neq \bzero$  equation  \eqref{eq:ine0_nl} can be written in the form
 \[
\frac{\partial_{ \bq_\sR}\psi}{ |\partial_{\bq_\sR} \psi|} \cdot { \dot\bq_\sR}
=- \frac{\theta }{ \rho_{_R}|\partial_{\bq_\sR} \psi|}\Big(\frac {\bq_{\sR}} {\theta^2} \cdot \nablaR \theta+\rho_\sR\zeta+\nablaR\cdot\bk_\sR\Big),
\]
and applying \eqref{eq:BNG} with $\bN=\partial_{\bq_\sR} \psi/|\partial_{\bq_\sR} \psi|$, $\bZ={{\dot\bq}_\sR}$ 
we obtain
\beq
{{\dot\bq}_\sR}=\Big( {\bq_{\sR}} \cdot \nablaR \frac1 \theta-\rho_\sR\zeta-\nablaR\cdot\bk_\sR\Big) \frac{\theta\bN }{ \rho_{_R}|\partial_{\bq_\sR} \psi|} + (\bone - \bN\otimes \bN)\bG ,
\label{Hypo_mat_GK}\eeq
where $\bG$ is an arbitrary vector-valued function dependent on $(\theta,\bq_\sR,\nablaR \theta)$. Now we let
\beq
\rho_{\sR}\psi=\rho_{\sR}\psi_0(\theta)+\frac{\tau\theta}{2\varkappa(\theta)}|\bq_\sR|^2, \quad 
\rho_\sR\zeta =\frac1{\varkappa(\theta)}|\bq_\sR|^2+\ell^2|\nablaR\bq_\sR|^2+2\ell^2|\nablaR\cdot\bq_\sR|^2, 
\label{eq:GK_energy} \eeq
\[
\bk_\sR=-\ell^2[\bq_\sR\nablaR\bq_\sR+2(\nablaR\cdot\bq_\sR)\bq_\sR].
\]
where $\ell$ is a constant parameter and $\varkappa$ a positive function.
Noting that
\[
[\nabla_\sR^2\bq_\sR+2\nablaR(\nablaR\cdot\bq_\sR)]\cdot\bq_\sR=\nablaR\cdot[\bq_\sR\nablaR\bq+2(\nablaR\cdot\bq_\sR)\bq_\sR]-
|\nablaR\bq_\sR|^2-2|\nablaR\cdot\bq_\sR|^2,
\]
the representation formula  \eqref{Hypo_mat_GK} gives
\[
{{\dot\bq}_\sR}= \bG+\frac1\tau\bN\otimes \bN\Big(\varkappa(\theta)\nablaR\frac 1\theta-\bq_\sR+\ell^2\varkappa(\theta)[\nabla_\sR^2\bq+2\nablaR(\nablaR\cdot\bq)]-\tau\bG\Big),
\]
where $\bN=\bq_\sR/|\bq_\sR|$. Finally, letting 
$$\tau\bG=\varkappa(\theta)\Big(\nablaR \frac1\theta-\bq_\sR+\ell^2[\nabla_\sR^2\bq+2\nablaR(\nablaR\cdot\bq)]\Big)$$  we obtain \eqref{eq:GrK} provided that 
\beq\lambda^2(\theta)=\ell^2\varkappa(\theta),\qquad\kappa(\theta)=\varkappa(\theta)/\theta^2.\label{parameters_dep}\eeq
These conditions reveal that temperature-dependent material parameters of the GK model are functionally connected. Notably, this can also be considered a non-trivial consequence of the {\it Onsagerian relation} (see, e.g., \cite[eqn.7]{RKFAJ}).

When $\tau$ is negligible then equation \eqref{eq:GrK} reduces to \eqref{eq:nlFourier}.
If this is the case, the free energy $\psi$ is independent of $\bq_\sR$. However, from the comparison of \eqref{eq:nl_F_diss} and \eqref{eq:GK_energy}  we observe that $\zeta$ and $\bk_\sR$ have the same expression, whether $\tau=0$ or $\tau>0$.  

\begin{proposition}\label{GKcond}
Since $\ell$ and $\varkappa(\theta)$ are undetermined, \eqref{eq:GrK} is thermodynamically consistent with the nonlocal statement \eqref{eq:ine0_nl} of the Second Law provided that $\kappa(\theta)$ is proportional to $\lambda^2(\theta)/\theta^2$ and the nonlinear boundary condition
\beq\bk_\sR\cdot\bn_\sR\Big\vert_{\partial\cR}=\bq_\sR\cdot\frac{\partial\bq_\sR}{\partial n_\sR}\Big\vert_{\partial\cR}+2(\nablaR\cdot\bq_\sR)\bq_\sR\cdot\bn_\sR\Big\vert_{\partial\cR}=0\label{eq:bc_GK}\eeq
is satisfied. 
\end{proposition}
In the seminal paper \cite{GK} steady one-dimensional  solutions are derived in a cylinder assuming $\bq_\sR=\bzero$ on its boundary, a condition that certainly implies \eqref{eq:bc_GK}. Furthermore,  \eqref{eq:bc_GK} is perfectly consistent with boundary conditions scrutinized in \cite{SXAJ} in connection with steady solutions in nanowires. 
It is worth noting that \eqref{eq:bc_GK} can be ensured by the sufficient condition
\beq
\bq_\sR\cdot\frac{\partial\bq_\sR}{\partial n_\sR}\Big\vert_{\partial\cR}=0, \qquad \bq_\sR\cdot\bn_\sR\Big\vert_{\partial\cR}=0,
\label{suff_bc_GK}\eeq
This condition allows an axially symmetric profile of heat flux in a cylindrical pipe, a profile reminiscent of Hagen-Poiseuille flow for viscous Newtonian fluids.


\subsubsection{Nonlinear Guyer-Krumhansl conductors}
A simple nonlinear extension of the Guyer-Krumhansl model can be written as\footnote{\,When $\beta=0$, see for instance  \cite{CSJ_2009,CSJ_2010}.}
\beq
\tau\dot \bq_\sR+\bq_\sR+\kappa \nablaR\theta=\lambda^2[\nabla_\sR^2\bq_\sR+2\nablaR(\nablaR\cdot\bq_\sR)]+\mu\bq_\sR\nablaR\bq_\sR+\nu(\nablaR\cdot\bq_\sR)\bq_\sR.
\label{eq:GrK_nl}\eeq
This model illustrates relaxational and nonlocal effects of the heat flow in nanosystems  \cite{CSJ_2010bis}.
Borrowing the procedure of the previous subsection, we choose $\psi$ and $\zeta$ as in \eqref{eq:GK_energy} and we let
\[
\bk_\sR=-\ell^2[\bq_\sR\nablaR\bq_\sR+2(\nablaR\cdot\bq_\sR)\bq_\sR]-\delta|\bq_\sR|^2\bq_\sR.
\]
First we compute $\nablaR\cdot\bk_\sR$ by exploiting the identity 
\[
\nablaR\cdot\big[|\bq_\sR|^2\bq_\sR\big]=2\bq_\sR\nablaR\bq_\sR\cdot\bq_\sR+|\bq_\sR|^2\nablaR\cdot\bq_\sR,
\]
then we apply the representation formula \eqref{Hypo_mat_GK} as in the previous section. As a result we recover the constitutive equation  \eqref{eq:GrK_nl} provided that
\beq\lambda^2(\theta)=\ell^2\varkappa(\theta), \qquad \kappa(\theta)=\varkappa(\theta)/\theta^2, \qquad \mu(\theta)=2\delta\varkappa(\theta),\qquad \nu(\theta)=\delta\varkappa(\theta).\label{eq:consist}\eeq
As before, $\varkappa(\theta)>0$ ensures that the internal entropy supply is nonnegative, but thermodynamics does not impose any restrictions on the sign of $\delta$.
\begin{proposition}
Since $\ell$, $\nu$ and $\varkappa(\theta)$ are undetermined, \eqref{eq:GrK_nl} is thermodynamically consistent with the nonlocal statement of the Second Law provided that 
\begin{itemize}
  \item [i) ] $\kappa(\theta)$ is proportional to $\lambda^2(\theta)/\theta^2$;
  \item [ii) ] $\mu(\theta)=2\nu(\theta)$ is proportional to $\lambda^2(\theta)$;
  \item [iii) ] $\bq_\sR$ satisfies the nonlinear boundary conditions
\beq\bq_\sR\cdot\frac{\partial\bq_\sR}{\partial n_\sR}\Big\vert_{\partial\cR}+2(\nablaR\cdot\bq_\sR)\bq_\sR\cdot\bn_\sR\Big\vert_{\partial\cR}=0,
\qquad |\bq_\sR|^2\bq_\sR\cdot\bn_\sR\Big\vert_{\partial\cR}=0\label{eq:bc_GK_nl}\eeq
\end{itemize} 
\end{proposition}
We can easily verify that \eqref{eq:bc_GK_nl}, as well as  \eqref{eq:bc_GK},  is ensured by the sufficient condition \eqref{suff_bc_GK}.

\begin{remark}\label{6_1}
Let
\[\bQ_\sR=\lambda^2[\nablaR\bq+2(\nablaR\cdot\bq_\sR)\bone]+\frac\mu 2|\bq|^2\bone+\nu \bq\otimes\bq\,,\]
and let all parameters take constant values. Then equation \eqref{eq:GrK_nl} can be rewritten as
\[
\tau\dot \bq_\sR+\bq_\sR+\kappa \nablaR\theta=\nablaR\cdot\bQ_\sR,
\]
that represents the evolution equation of the heat flux in the framework of extended irreversible thermodynamics \cite{JCVL,Van2015}. In particular, using \eqref{eq:consist} with $\varkappa(\theta)=\varkappa_0$, we obtain
\[\bQ_\sR=\varkappa_0\big(\ell^2[\nablaR\bq+2(\nablaR\cdot\bq_\sR)\bone]+\nu\big[|\bq|^2\bone+\bq\otimes\bq\big]\big).\]
Note that $\bQ_\sR$ is related to the extra entropy flux $\bk_\sR$ by the relation
$$\bQ_\sR=-\varkappa_0\,\partial_{\bq_\sR}\bk_\sR.$$
\end{remark}
\section{Conclusions}
This paper is devoted to develop a general constitutive scheme within continuum thermodynamics to describe the behavior of heat flow in deformable media. Starting from a classical thermodynamic approach in the material (Lagrangian)  description, a wide class of rate-type constitutive equations for the heat flux vector $\bq_\sR$ are obtained. Their thermodynamic consistency with the Second Law in the local form  \eqref{eq:CD} is established by exhibiting functional expressions of the specific free energy $\psi$ and entropy production rate $\sigma$. Some nonlinear anisotropic models are generated by means of a suitable Representation Lemma. For instance, heat conduction models with zero dissipation \eqref{Hypo_mat_2}, nonlinear Maxwell-Cattaneo-Vernotte-like equations \eqref{Hypo_mat_CM} and anisotropic equations of the Jeffreys type \eqref{eq:Jeffreys}. For each of these, 
restrictions on the material parameters are derived to ensure non-negativity of entropy production.

Linear models involving the notion of ``thermal displacement" (such as Green-Naghdi and Quintanilla constitutive equations) are reformulated here as equations of the rate-type, namely \eqref{eq:GNIII} and \eqref{eq:Q}, and their free energy and entropy production are explicitly determined as quadratic functions. When Quintanilla's model is considered, such quadratic functions are non trivial;
we also note that the necessary and sufficient condition on the material parameters for entropy production $\sigma$  to be non-negative (see Proposition \ref{qprop}) is equivalent to the estimate $iii)$ of \cite{Quintanilla} which guarantees the stability of the thermoelastic problem.
A new linear model, inspired by the constitutive equation of a Burgers fluid, is proposed  considering a mixture of two different substances both characterized by a Maxwell-Cattaneo type heat exchange mechanism.  This results in the rate equation \eqref{JCL} involving the second time derivative of the heat flux; its thermodynamic consistency is proved in Proposition \ref{Burcon}. For rigid bodies the corresponding third order temperature equation \eqref{JP_temp_eq}  was first proposed by Joseph and Preziosi in \cite{JP} based on the linearization of the Gurtin-Pipkin model.

Finally, linear and nonlinear Guyer-Krumhansl heat conduction models are derived using the Representation Lemma.
Their thermodynamic consistency with the Second Law in the nonlocal form \eqref{basic_ent_ineq} is established by exhibiting explicit expressions of the specific free energy $\psi$, entropy production rate $\zeta$ and extra entropy flux $\bk_\sR$. 
We believe this is the first demonstration that the Guyer-Krumhansl equation can be derived in the framework of classical thermodynamics provided that
material parameters depend on temperature. It is worth noting that necessary and sufficient conditions for entropy production $\zeta$  to be non-negative  imply their mutual functional dependence (see \eqref{parameters_dep} and \eqref{eq:consist} in the linear and nonlinear case, respectively).
Notably, this dependence can also be obtained within different frameworks (see, e.g., \cite{RKFAJ}).
Restrictions imposed on $\bq$ and its gradient by the no-flow boundary condition \eqref{strong_II_law}  for $\bk_\sR$ are discussed. The resulting analysis carries important suggestions for the choice of the most appropriate boundary conditions in applications.
In Remark \ref{6_1}, a connection with the evolution equation of the heat flux in the framework of extended irreversible
thermodynamics is established  (see \cite{JCVL,Van2015}).

A final remark is in order. In connection with the same heat conduction model (local or non-local) different schemes are possible.
As shown with the Jeffreys-type equation \eqref{eq:Jeffreys}, for instance, different free energy functions $\psi$ and entropy productions $\sigma$ are allowed with the same conclusion (see Remark \ref{non_uniq}). This in turn means that the free energy and the entropy production need not be unique for a given rate equation.


\end{document}